\documentclass{article}

\usepackage[margin=1in]{geometry}
\usepackage[T1]{fontenc}
\usepackage{amsmath, amssymb, amsthm, mathtools}
\usepackage[cal=euler]{mathalpha}
\usepackage{cochineal}
\usepackage{microtype}
\usepackage{graphicx}
\usepackage[numbers,sort&compress]{natbib}
\usepackage{hyperref}
\usepackage{booktabs}
\usepackage{xcolor}

\newcommand{\akshay}[1]{}
\newcommand{\vac}[1]{}

\newtheorem{theorem}{Theorem}[section]
\newtheorem{proposition}[theorem]{Proposition}
\newtheorem{lemma}[theorem]{Lemma}
\newtheorem{corollary}[theorem]{Corollary}

\newcommand{\Div}{\mathrm{D}}
\newcommand{\Ent}{\mathrm{H}}
\newcommand{\TC}{\mathrm{TC}}
\newcommand{\JS}{\mathrm{JS}}
\newcommand{\EV}{\mathbb{E}}
\newcommand{\Prob}{\text{\upshape Pr}}
\newcommand{\cX}{\mathcal{X}}
\newcommand{\cA}{\mathcal{A}}
\newcommand{\cH}{\mathcal{H}}
\newcommand{\cT}{\mathcal{T}}
\newcommand{\cF}{\mathcal{F}}
\newcommand{\cI}{\mathcal{I}}
\newcommand{\Pemp}{\hat P_n}
\newcommand{\dkl}[2]{\Div(#1\,\|\,#2)}
\newcommand{\ip}[2]{\langle #1,\,#2\rangle}
\newcommand{\ind}{\mathbf{1}}
\newcommand{\om}{\omega_{\cA}}
\newcommand{\muA}{\mu_{\cA}}
\newcommand{\Qs}{Q^{\ast}}
\DeclareMathOperator{\supp}{supp}
\DeclareMathOperator*{\argmin}{arg\,min}

\makeatletter
\providecommand{\bvapx@title}{}
\AtBeginDocument{\@ifundefined{@title}{}{\global\let\bvapx@title\@title}}
\newcommand{\appendixtitle}{%
  \clearpage
  \phantomsection
  \begin{center}
    {\LARGE\bfseries Appendices to ``\bvapx@title''\par}%
  \end{center}
  \vspace{1.5\baselineskip}%
}
\makeatother

\title{An information identity reveals the geometry of deviation events}
\author{Akshay Balsubramani \\ {\small \texttt{akshay@vac.bio}}}
\date{}

\providecommand{\akshay}[1]{}\renewcommand{\akshay}[1]{}%
\providecommand{\vac}[1]{}\renewcommand{\vac}[1]{}%
\begin{document}
\maketitle

\begin{abstract}
In a deviation event, the empirical measure of $n$ independent draws lands in a set of distributions.
The exponent of its probability splits, at every $n$, into two nonnegative terms.
The first is the rate: $n$ times the relative entropy, from the population, of the distribution of a typical draw under the event.
The second is the dependence that conditioning induces among the draws, measured by their total correlation.
Their relative sizes depend on the geometry of the set.
This paper focuses on this exact split and on the geometry that the interplay of its two terms reveals.
The dependence vanishes exactly when the event confines every draw to a single set, and the classical rate of that constraint is exact.
On a set invariant under a compact group that fixes the population and leaves no invariant set of intermediate probability, the rate vanishes; uniform convergence over a hypothesis class that such a group permutes is one.
On a half-space whose threshold sits a fixed number of standard errors above the mean, the fraction of the exponent that is dependence tends to a function of the event's probability alone.
On a set split into pieces, the dependence is the pieces' average plus $n$ times the Jensen--Shannon divergence of their marginals, minus the entropy of their weights.
The split extends to a general reference law.
When the information projection is such a law, the split gives the exact value of the expectation in the exponential change of measure to the projection.
\end{abstract}

\section{Introduction}
\label{sec:intro}

Sanov's theorem gives the exponential rate at which the empirical measure $\Pemp$ of $n$ independent draws from a population $P$ falls into a set $\cA$ of distributions: the relative entropy from $P$ of the information projection of $P$ onto $\cA$ \cite{Sanov57,csiszar1984Sanov,dembo2010large}.
The finite-sample inequalities in daily use read this rate off at finite $n$: Chernoff's and Hoeffding's bounds for a single statistic, the union bound over a hypothesis class, and the Donsker--Varadhan change of measure behind PAC--Bayes bounds \cite{hoeffding1963,cover1999elements,vapnik1971,hellstrom2025generalizationpacbayes}.
Chernoff's bound reads the rate off with no correction, the union bound and the change of measure incur a polynomial or logarithmic one, and Hoeffding's bound relaxes the rate itself.
When performing statistical inference, the rate is read off the same way with an e-variable, a nonnegative statistic whose mean under the null hypothesis, here $P$, is at most one.
On a closed convex set that excludes $P$, the likelihood ratio of the information projection is, under regularity conditions, the growth-optimal e-variable, and Markov's inequality applied to it gives the Csisz\'ar--Sanov--Chernoff bound \cite{grunwald2024growthoptimal}.
These exact readings rest on the convexity of the set, through the Pythagorean inequality for information projections \cite{csisz1975idivergenceb}.
Convexity is also their obstacle: the events of uniform convergence over a class and of a two-sided tail are unions of half-spaces, which are generally not convex.

An exact form needs no convexity or closedness of $\cA$, and no structure on the space beyond measurability.
Write $\muA$ for the law of the sample conditioned on $\Pemp\in\cA$ and $\om$ for its one-coordinate marginal, the distribution of a typical draw under the event.
The total correlation of $\muA$ is $\TC(\cA) = \dkl{\muA}{\om^{\otimes n}}$, its relative entropy from the product of its own marginals \cite{watanabe1960information}; it depends on $P$ only through $\muA$.
Then at every $n$
\begin{equation}
\label{eq:lemma4}
  -\log \Prob_P(\Pemp \in \cA) \;=\; \underbrace{n\,\dkl{\om}{P}}_{\text{rate}} \;+\; \underbrace{\dkl{\muA}{\om^{\otimes n}}}_{\text{dependence}}
\end{equation}
with both terms nonnegative.
The first is the rate of the conditioned marginal, and the second is the dependence that conditioning induces among the coordinates. 
Replacing $\om$ by the information projection needs $\om\in\cA$, which convexity of $\cA$ guarantees on a finite alphabet \cite{balsubramani2020sharp}.
The exponential change of measure to the projection factors out the Sanov rate exactly and leaves an expectation over a sample from the projection \cite{cram1938,BahadurRao1960,ney1983dominating,dinwoodie1992mesures}.
Theorem~\ref{thm:identity} writes the same exponent at a general reference law.
At the reference $\om$ it is \eqref{eq:lemma4}.
When the projection $\Qs$ is a reference law and $P\ll\Qs$ (Section~\ref{sec:identity}), the theorem at the projection gives the exact value of that expectation.
Section~\ref{sec:related} gives the sources of \eqref{eq:lemma4}.

The relative sizes of the two terms depend on the geometry of $\cA$ (Figure~\ref{fig:regime-dial}).
It is instructive to consider two extreme cases.
At one extreme is a half-space, where the first term is the exponent up to a logarithm in $n$.
There the coordinates decouple as $n$ grows: any fixed number of them become independent in the limit, with common law the information projection.
This limit is Gibbs conditioning, the conditional limit theorem that accompanies Sanov's theorem.
The second term, the dependence of the whole sample, does not vanish in this limit.
Section~\ref{sec:special-cases} gives an example in which it rises with $n$, at most logarithmically, while the dependence of any two coordinates falls to zero.
At the other extreme is a union of half-spaces permuted by a symmetry that fixes $P$ and acts transitively on the alphabet.
Its conditioned marginal is $P$ itself, so the first term is zero and every nat of the exponent is dependence among the sampled points.
The rate of its most probable type still lies within $K\log(n+1)$ of the exponent, for an alphabet of $K$ letters, while the mechanism that rate suggests, a drift of the marginal toward the projection, is absent.

The dependence vanishes if and only if the event confines every draw to one measurable set, and the classical rate of that constraint is exact at every $n$ (Proposition~\ref{prop:face}).
On a discrete alphabet with $P$ of full support these are the events whose types form a face of the simplex.
The marginal rate vanishes on every set invariant under a compact group that fixes $P$ and leaves no invariant set of intermediate probability (Proposition~\ref{prop:symmetry}), so the whole exponent is dependence there.
Uniform convergence over a hypothesis class that such a group permutes is such an event, and so is a large empirical mean of uniform points on a sphere.
The dependence of a set split into pieces is the within-piece dependence plus $n$ times the Jensen--Shannon divergence of the piece marginals minus the entropy of the piece weights (Theorem~\ref{thm:partition}), so an event made of separated modes of comparable probability has dependence of order $n$.
Section~\ref{sec:results} states the identity and these results, which hold on any measurable space. 

Section~\ref{sec:finite} adds one hypothesis, unrelated to convexity: a finite alphabet with $P$ of full support.
On a finite alphabet the method of types supplies an envelope the identity does not.
On every set the dependence lies within $K\log(n+1)$ of $n$ times the gap between the rate of the most probable type and the rate of the conditioned marginal (Proposition~\ref{prop:bracket}).
On a half-space whose threshold sits a fixed number of standard errors above the mean, the fraction of the exponent that is dependence tends to a function of the event's probability alone, the same for every population and every nonconstant statistic (Proposition~\ref{prop:halfspace-gauss}).
On a convex set the log-probability gap against the tangent half-space at the projection is nonnegative at every $n$ and zero only when the two sets have the same types (Proposition~\ref{prop:tangent}).
On a curved set the gap converges to a curvature constant as $n$ grows, unless the log-likelihood ratio of the projection to the population takes its values on an arithmetic progression (Proposition~\ref{prop:curvature-limit}).
In that case the gap follows the position of the types at sample size $n$ relative to the half-space.
The gap has no limit unless that position is the same at every $n$.
Section~\ref{sec:special-cases} reads the classical results off the identity.
Section~\ref{sec:learning} finds the dependence share in the deviation events of learning.

\section{Setup}
\label{sec:setup}
\subsection{The conditioned sample and its total correlation}
\label{sec:setup-sample}

The sample $Z = (Z_1, \ldots, Z_n)$ is independent with law $P$, and its empirical measure is $\Pemp = \frac1n\sum_i \delta_{Z_i}$.
The space $\cX$ is measurable, and the distributions on it are given the $\sigma$-algebra generated by the evaluations $t\mapsto t(B)$ over measurable $B\subseteq\cX$, so that $\Pemp$ is measurable and $E_\cA$ below is an event.
For a measurable set $\cA$ of distributions on $\cX$ with $\Prob_P(\Pemp\in\cA) > 0$, the deviation event is $E_\cA = \{z \in \cX^n : \hat P_n(z) \in \cA\}$, the conditioned sample law is $\muA = P^{\otimes n}(\cdot \mid E_\cA)$, and $\om$ is its one-coordinate marginal, the same along every coordinate by exchangeability and equal to the barycenter $\EV_{\muA}[\Pemp]$.
Exchangeability is automatic because $E_\cA$ and $P^{\otimes n}$ are permutation invariant, and $\om\ll P$ because $\muA\ll P^{\otimes n}$.
The total correlation of $\muA$ is $\TC_P(\cA) := \dkl{\muA}{\om^{\otimes n}}$.
It is nonnegative, and zero exactly when $\muA$ is a product measure \cite{watanabe1960information}.
On a discrete alphabet with $\Ent(\om)$ finite it is the multi-information $n\,\Ent(\om) - \Ent(\muA)$; the divergence is the form that survives elsewhere, since entropies need not be finite while $\TC_P(\cA)$ is bounded by the exponent.
The subscript records that the population enters through $\muA$.
Every total correlation below is taken under $P$, so the subscript is dropped.

\subsection{The information projection and the defect}
\label{sec:setup-projection}

The information projection of $P$ onto $\cA$ is $\Qs \in \argmin_{Q\in\cA}\dkl{Q}{P}$, unique when $\cA$ is convex and closed in variation \cite{csisz1975idivergenceb}.
The inner product of a measurable function $f:\cX\to[-\infty,\infty)$ with a distribution $t$ on $\cX$ is $\ip{f}{t} = \ip{t}{f} := \int f\,dt$, whenever the integral is defined.
For two distributions $t$ and $s$, $\ip{f}{t - s} := \ip{f}{t} - \ip{f}{s}$ whenever the difference is defined.
At an empirical measure, $\ip{f}{\Pemp} = \frac1n\sum_i f(Z_i)$.
When $\Qs\ll P$, fix a version of $g := \log(d\Qs/dP)$, taking the value $-\infty$ where the density vanishes; the defect of a distribution $t$ on $\cX$ is
\begin{equation}
\label{eq:defect}
  \Delta(t) := \ip{g}{t - \Qs}
\end{equation}
which equals $\dkl{t}{P} - \dkl{\Qs}{P} - \dkl{t}{\Qs}$ whenever $\dkl{t}{P}$, $\dkl{t}{\Qs}$ and $\dkl{\Qs}{P}$ are finite (Lemma~\ref{lem:defect-lr}(i)).
The defect is defined by the inner product because an empirical measure on a nonatomic space has $\dkl{\Pemp}{P} = \infty$, while $\ip{g}{\Pemp}$ is almost surely finite for a sample from $\Qs$.
The Pythagorean inequality for information projections \cite{csisz1975idivergenceb} is the statement $\Delta \ge 0$ at every $t$ of a convex $\cA$ with $\dkl{t}{P} < \infty$, with equality throughout a linear family, a set of distributions cut out by finitely many linear equalities.

\subsection{The exponential change of measure}
\label{sec:setup-tilt}

Equation~\eqref{eq:lemma4} holds for every $\cA$.
An exponential change of measure gives another exact form of the deviation probability.
When $P$ and $\Qs$ are mutually absolutely continuous, $g$ is finite $P$-almost everywhere and $dP/d\Qs = e^{-g}$.
The likelihood ratio of the sample is then the product $dP^{\otimes n}/d(\Qs)^{\otimes n} = \prod_{i=1}^n e^{-g(Z_i)} = e^{-n\ip{g}{\Pemp}}$, a function of the empirical measure alone (Lemma~\ref{lem:defect-lr}(ii)).
Weighting a sample from $\Qs$ by this ratio gives, at every $n$ and for every $\cA$,
\begin{equation}
\label{eq:sharp-sanov}
  -\log\Prob_P(\Pemp\in\cA) = n\,\dkl{\Qs}{P} - \log \EV_{\Qs}\!\bigl[\ind\{\Pemp\in\cA\}\,e^{-n\Delta(\Pemp)}\bigr]
\end{equation}
with the expectation over a sample from $\Qs$.
It is classical on a half-space and at a dominating point or measure, where the defect is nonnegative on $\cA$ \cite{cram1938,BahadurRao1960,dembo2010large,ney1983dominating,dinwoodie1992mesures}. 

\begin{figure}[!t]
\centering
\includegraphics[width=\linewidth]{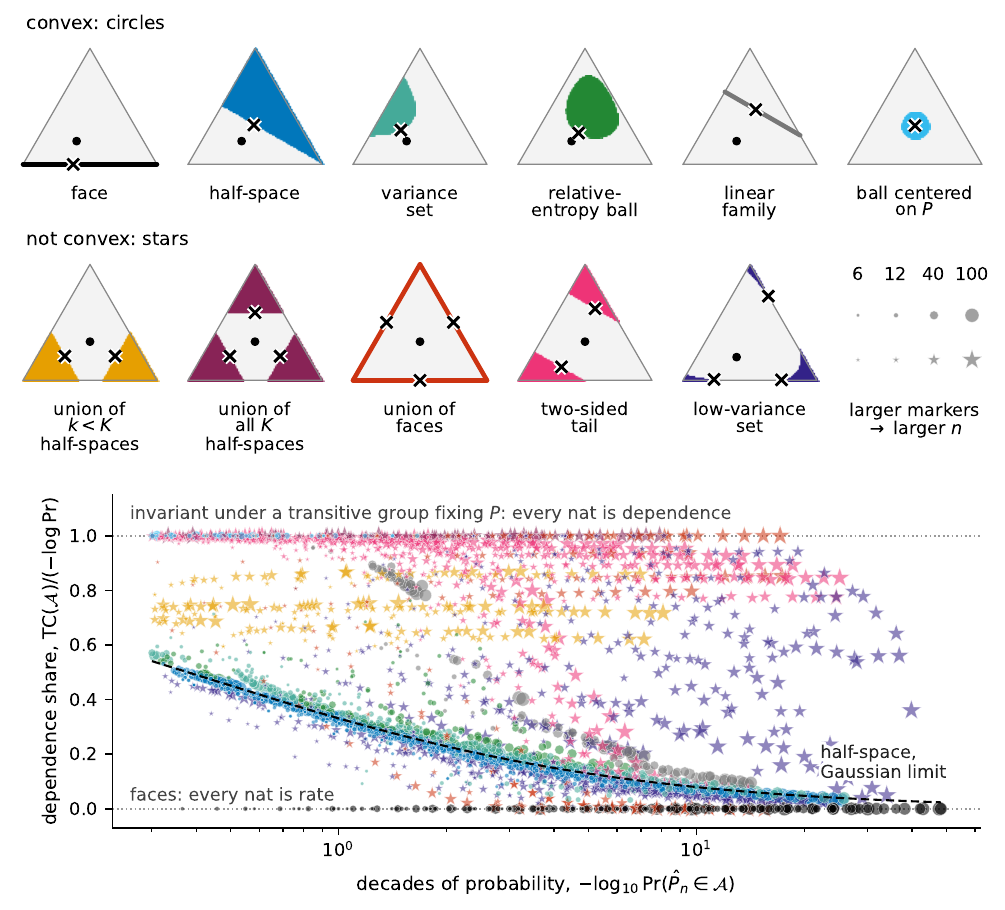}
\caption{\textbf{The fraction of the exponent that is dependence runs from zero on a face to one on a set invariant under a transitive group that fixes $P$, whether or not that set is convex.}
Top: one example set of each family's form on the simplex of three letters, the convex families in the first row, with the population $P$ (dot) and the information projection onto each piece of the set (crosses).
Bottom: the dependence share $\TC(\cA)/(-\log\Prob(\Pemp\in\cA))$ against the rarity of the event, for configurations of these families on alphabets of three and four letters at sample sizes from $6$ to $100$.
Circles mark the convex families and stars the others.
Larger markers mark larger $n$.
The markers are translucent, so the color deepens where they overlap.
The dependence vanishes on the faces and on no other configuration.
The marginal rate vanishes exactly on the configurations whose types at that $n$ form a set invariant under a transitive group that fixes $P$.
They are, on a uniform $P$, the unions of all $K$ half-spaces or all $K$ faces, the balls centered on $P$, and the low-variance sets whose types at that $n$ are the vertices of the simplex.
On the half-spaces, variance sets and linear families the share falls as the event gets rarer.
The dashed curve is the limit of a half-space's share when its threshold sits a fixed number of standard errors above the mean (Proposition~\ref{prop:halfspace-gauss}), a function of the probability alone.
The half-spaces lie below it or barely above it and approach it as $n$ grows.
On a uniform population the dependence share of the unions of fewer than $K$ half-spaces or faces and of the two-sided tails is above one half.
On a skewed population the most probable face comes to dominate a union of faces, so the share of the union falls toward zero.}
\label{fig:regime-dial}
\end{figure}

\section{Results on a measurable space}
\label{sec:results}

The results of this section hold on the measurable space of Section~\ref{sec:setup}, with no convexity or closedness of $\cA$.

\subsection{An information-theoretic identity}
\label{sec:identity}

Equation~\eqref{eq:lemma4} writes the exponent at the conditioned marginal $\om$, and \eqref{eq:sharp-sanov} writes it at the projection $\Qs$.
A reference law is a distribution $R$ on $\cX$ with $\dkl{R}{P}$ and $\dkl{\om}{R}$ finite, so that $R\ll P$ and $\om\ll R$.
One identity writes the exponent at every reference law.
At $R = \om$ it is \eqref{eq:lemma4}.
When the projection $\Qs$ is itself a reference law and $P\ll\Qs$, the identity at $R = \Qs$, compared with \eqref{eq:sharp-sanov}, gives the exact value of the expectation in \eqref{eq:sharp-sanov}.

\begin{theorem}[One identity, any reference]
\label{thm:identity}
Let $\cA$ satisfy $\Prob_P(\Pemp\in\cA) > 0$, and let $R$ be a distribution on $\cX$ with $R\ll P$ and $\om\ll R$, with $\dkl{R}{P}$ and $\dkl{\om}{R}$ both finite; set $g := \log(dR/dP)$.
Then
\begin{equation}
\label{eq:identity}
  -\log\Prob_P(\Pemp\in\cA) = n\,\dkl{R}{P} + n\,\ip{g}{\om - R} + n\,\dkl{\om}{R} + \TC(\cA)
\end{equation}
\end{theorem}

The four terms are the rate of the reference, two transport terms and the dependence.
The transport terms are the first-order and second-order costs of moving the reference to the conditioned marginal.
The dependence is the only term that does not depend on the reference.
Theorem~\ref{thm:identity} at $R = \Qs$ writes the exponent as $n\,\dkl{\Qs}{P} + n\,\Delta(\om) + n\,\dkl{\om}{\Qs} + \TC(\cA)$.
Equation~\eqref{eq:sharp-sanov} writes it as $n\,\dkl{\Qs}{P}$ minus the logarithm of the expectation over a sample from $\Qs$.
Corollary~\ref{cor:two-references}(ii) equates the two.

\begin{corollary}[Two references]
\label{cor:two-references}
Under the hypotheses of Theorem~\ref{thm:identity}:
\begin{enumerate}
  \item[(i)] at $R = \om$ both transport terms vanish and \eqref{eq:identity} is \eqref{eq:lemma4};
  \item[(ii)] if $\cA$ has an information projection $\Qs$ with $P$ and $\Qs$ mutually absolutely continuous, at $R = \Qs$ the first-order term is $n\,\Delta(\om)$ and
  \begin{equation}
  \label{eq:three-way}
    -\log\EV_{\Qs}\!\bigl[\ind\{\Pemp\in\cA\}\,e^{-n\Delta(\Pemp)}\bigr]
    = n\,\dkl{\om}{\Qs} + n\,\Delta(\om) + \TC(\cA)
  \end{equation}
\end{enumerate}
\end{corollary}

Equation~\eqref{eq:three-way} is the exact difference between \eqref{eq:lemma4} and \eqref{eq:sharp-sanov}.
The defect term $n\,\Delta(\om)$ vanishes on a linear family.
It is nonnegative on a convex set that contains $\om$, as every convex set on a finite alphabet does.
On a half-space that excludes $P$ and is not a face it is positive at every $n$, because the barycenter lies strictly inside the half-space while the projection lies on its boundary.
Off convexity it can be negative.
On a union of half-spaces whose barycenter is $P$, as when a transitive symmetry permutes the pieces, it equals $n\,\ip{g}{P - \Qs} = -n\,[\dkl{P}{\Qs} + \dkl{\Qs}{P}]$, because $\ip{g}{P} = -\dkl{P}{\Qs}$ and $\ip{g}{\Qs} = \dkl{\Qs}{P}$.

\subsection{Faces and symmetric sets}
\label{sec:taxonomy}

By \eqref{eq:lemma4} the dependence share $\TC(\cA)/(-\log\Prob_P(\Pemp\in\cA))$, the fraction of the exponent that is dependence, lies in $[0,1]$ for every event of probability below one.
The share is zero on a face and one on a set invariant under a transitive group that fixes $P$ (Figure~\ref{fig:regime-dial}).

\begin{proposition}[Zero dependence]
\label{prop:face}
$\TC(\cA) = 0$ if and only if $E_\cA = F^n$ up to a $P^{\otimes n}$-null set, for some measurable $F\subseteq\cX$ with $P(F) > 0$.
In that case $\muA = P(\cdot\mid F)^{\otimes n}$, $\om = P(\cdot\mid F)$ and $-\log\Prob_P(\Pemp\in\cA) = -n\log P(F)$.
If $\cA$ is the face $\{Q : Q(F) = 1\}$ itself, then also $\Qs = \om$, $\Delta\equiv 0$ on $\cA$, and the classical rate $n\,\dkl{\Qs}{P}$ equals the exponent at every $n$.
\end{proposition}

A face event conditions every point to land in $F$, and points conditioned separately stay independent.
Zero training loss of a predictor fixed before the sample, which confines every point to the set that predictor classifies correctly, is a face event.

\begin{proposition}[Symmetric events]
\label{prop:symmetry}
Let a compact group $G$ act measurably on $\cX$ with induced action on distributions, let $P$ and $\cA$ be $G$-invariant, and let $\cI$ be the $\sigma$-algebra of $G$-invariant sets.
Then $\om$ is $G$-invariant and $\dkl{\om}{P} = \dkl{\om|_\cI}{P|_\cI}$.
If $\cI$ is $P$-trivial then $\om = P$ and $-\log\Prob_P(\Pemp\in\cA) = \TC(\cA)$.
\end{proposition}

The marginal rate of a symmetric set sees only the invariant sets.
For a finite group on a discrete alphabet these are the unions of orbits, and the restriction to $\cI$ is the distribution of orbit masses.
The triviality of $\cI$ holds when $P$ is concentrated on a single orbit, which a transitive action gives at once; on a discrete alphabet the two are equivalent, and where $P$ has full support both amount to transitivity.
Triviality is the weaker hypothesis and the one that survives on a continuous space, where the rotations of a sphere act transitively and a smaller compact group can still leave $P$ concentrated on one orbit.
Under a transitive action the whole exponent is dependence at every $n$; the inequality $\dkl{\om}{P}\ge\dkl{\Qs}{P}$ of Theorem 3 in \cite{balsubramani2020sharp}, which needs $\om\in\cA$, fails by the full Sanov rate whenever $P\notin\cA$.

\subsection{The size of the dependence term}
\label{sec:union}

\begin{theorem}[Partition identity]
\label{thm:partition}
Let $\cA = \cA_1 \sqcup \cdots \sqcup \cA_k$ with $\Prob_P(\Pemp\in\cA_j) > 0$ for each $j$, and let $w_j := \Prob_P(\Pemp\in\cA_j \mid \Pemp\in\cA)$.
Then $\om = \sum_j w_j\,\omega_{\cA_j}$ and
\begin{equation}
\label{eq:partition}
  \TC(\cA) = \sum_{j=1}^k w_j\,\TC(\cA_j) + n\,\JS_w(\omega_{\cA_1},\ldots,\omega_{\cA_k}) - \Ent(w)
\end{equation}
where $\JS_w(\omega_1,\ldots,\omega_k) := \sum_j w_j\,\dkl{\omega_j}{\sum_i w_i\omega_i}$ and $\Ent(w) = -\sum_j w_j\log w_j$.
\end{theorem}

The label $J$ of the piece containing $\Pemp$ is a function of the sample, and \eqref{eq:partition} decomposes the total correlation through it.
Its three terms are the dependence within the pieces, the information $n\,\JS_w$ that the coordinates have about $J$ one at a time, and the information $\Ent(w)$ that the whole sample has about $J$.
On a finite alphabet this is the known decomposition of a total correlation through a conditioning variable \cite{versteeg2014discovering}, taken at $J$; the proof works with relative entropies and needs no finite alphabet.
The middle term is of order $n$ when the piece marginals stay a fixed distance apart and the weights stay bounded away from zero.

\section{Results on a finite alphabet}
\label{sec:finite}

Everything to this point holds on a measurable space.
The results of this section add one hypothesis, which stands for the rest of the paper: $\cX$ has $K = |\cX|$ letters and $P$ has full support.
The empirical measure $\Pemp$ then lies in the set $\cT_n = \{t : nt \in \mathbb{Z}_{\ge 0}^{\cX}\}$ of types at sample size $n$.
No other hypothesis stands for the whole section.
Only Corollary~\ref{cor:convex-log} and Propositions~\ref{prop:tangent} and~\ref{prop:curvature-limit} assume $\cA$ convex, and Corollary~\ref{cor:mixture-limit} assumes each of its pieces convex.
On a finite alphabet a function and a signed measure are both vectors indexed by the letters, and for two such vectors $f$ and $v$ with $f$ finite, $\ip{f}{v} = \sum_x f(x)\,v(x)$ is the Euclidean inner product.

\subsection{The method-of-types envelope}
\label{sec:envelope}

A most probable type in $\cA$ is $\Qs_n \in \argmin_{t \in \cA\cap\cT_n}\dkl{t}{P}$.

\begin{proposition}[Dependence bracket]
\label{prop:bracket}
For every $\cA$ with $\Prob_P(\Pemp\in\cA)>0$ and $\Qs_n$ a most probable type in $\cA$,
\begin{equation}
\label{eq:bracket}
  \bigl|\,\TC(\cA) - n\,[\dkl{\Qs_n}{P} - \dkl{\om}{P}]\,\bigr| \le K\log(n+1)
\end{equation}
\end{proposition}

\begin{corollary}[Convex events]
\label{cor:convex-log}
If $\cA$ is convex with information projection $\Qs$, then $\dkl{\om}{P}\ge\dkl{\Qs}{P}$ and
\[
  \TC(\cA) \le n\,[\dkl{\Qs_n}{P} - \dkl{\Qs}{P}] + K\log(n+1)
\]
\end{corollary}

The dependence is $n$ times the gap between the rate of the most probable type and the rate of the conditioned marginal, to within $K\log(n+1)$, the logarithm of the method-of-types count of at most $(n+1)^K$ types \cite{cover1999elements}.
That count enters the proof as a hypothesis, so the envelope is inherited from the method of types and not read off the identity.
On a convex set the marginal is feasible, so the gap is at most $\dkl{\Qs_n}{P} - \dkl{\Qs}{P}$, the error of approximating the projection by a type.
Off convexity the marginal rate can lie below the projection's by an amount that does not shrink with $n$, so the dependence is then of order $n$.
The same bound puts the rate $n\,\dkl{\Qs_n}{P}$ of the most probable type within $K\log(n+1)$ of the exponent on every set, whatever the dependence share.

\begin{corollary}[Mixture limit]
\label{cor:mixture-limit}
Let $\cA = \cA_1\sqcup\cdots\sqcup\cA_k$ with each $\cA_j$ convex, with information projection $\Qs_j$ of full support and most probable type $\Qs_{j,n}$, and suppose $n\,[\dkl{\Qs_{j,n}}{P} - \dkl{\Qs_j}{P}] \le c$ for all $n$ and $j$.
Then $\dkl{\omega_{\cA_j}}{\Qs_j} \le (c + K\log(n+1))/n$ and $\TC(\cA_j) \le c + K\log(n+1)$ for each $j$, and
\[
  \frac{\TC(\cA)}{n} - \JS_{w(n)}(\Qs_1,\ldots,\Qs_k) \longrightarrow 0
\]
\end{corollary}

The weights $w(n)$ are the conditional probabilities of the pieces.
A piece whose projection has a rate above the minimum receives exponentially small weight, so the Jensen--Shannon term concentrates on the modes of minimal rate, whose relative weights are set by the sub-exponential prefactors of their probabilities.
For a half-space piece whose statistic takes values on an arithmetic progression and whose threshold some type attains, the prefactor is the Bahadur--Rao constant for such a statistic \cite{BahadurRao1960}.
In the tilt $\lambda_j$ and the tilted standard deviation $\sigma_j$ of the piece's statistic, measured in units of the step of the progression, the constant is $1/((1-e^{-\lambda_j})\sigma_j)$.
Two such modes of equal rate then receive weights in the ratio of these constants (Appendix~\ref{sec:eval-union-tax}).
The equal-weight mixture over the projections asserted for finite alphabets \cite{grendar2006conditional} holds when a symmetry of the event permutes the modes and equalizes their prefactors.
The asymptotic dependence share is the Jensen--Shannon divergence of the minimal-rate modes at their limiting weights divided by the Sanov rate.

A union of half-spaces is the deviation event of uniform convergence over a finite class (Section~\ref{sec:learning}); its projection is the projection of a piece of smallest rate, and its Sanov rate is that rate.
When a symmetry that fixes $P$ and acts transitively on the alphabet permutes the pieces, Proposition~\ref{prop:symmetry} gives $\om = P$ and the whole exponent is dependence.
If the symmetry permutes the pieces transitively and no type lies in two pieces, their marginals are permutations of one another, $\JS_w = \dkl{\omega_{\cA_1}}{P}$, and \eqref{eq:partition} reduces to $-\log\Prob_P(\Pemp\in\cA) = -\log\Prob_P(\Pemp\in\cA_1) - \log k$, so the union bound is exact and the marginal rate of one piece has become dependence of the union.

Figure~\ref{fig:regime-signatures} in Appendix~\ref{sec:eval-regime-map} separates the two terms over the configurations of Figure~\ref{fig:regime-dial} with $n \le 40$ on three letters and $n \le 24$ on four.
On a convex set that excludes $P$ the dependence grows on a shallower slope than the exponent, within the bound of Corollary~\ref{cor:convex-log}, while on a union of separated modes of comparable probability it grows in proportion to the exponent, through the Jensen--Shannon term of Theorem~\ref{thm:partition}.
Where one mode comes to dominate, as on a union of faces of a skewed population, the dependence eventually falls instead.

\subsection{The half-space}
\label{sec:halfspace}

In the limit of large samples, a half-space whose threshold sits a fixed number of standard errors above the mean splits its exponent in a way that depends only on the probability of the event.
Write $\mu$ and $\sigma^2$ for the mean and variance of a statistic $\ell$ under $P$, $\bar\Phi$ for the standard normal tail, and $m(z) := e^{-z^2/2}/(\sqrt{2\pi}\,\bar\Phi(z))$.

\begin{proposition}[The half-space in the Gaussian regime]
\label{prop:halfspace-gauss}
Let $\sigma^2 > 0$, $z\in\mathbb{R}$ and $\cA_n := \{t : \ip{\ell}{t}\ge\mu + z\sigma/\sqrt n\}$.
As $n\to\infty$,
\begin{equation}
\label{eq:halfspace-gauss}
  -\log\Prob_P(\Pemp\in\cA_n)\to\log\frac{1}{\bar\Phi(z)},\qquad
  n\,\dkl{\omega_{\cA_n}}{P}\to\frac{m(z)^2}{2},\qquad
  \TC(\cA_n)\to\log\frac{1}{\bar\Phi(z)} - \frac{m(z)^2}{2}
\end{equation}
\end{proposition}

With $p = \bar\Phi(z)$ the limiting probability, the dependence share of $\cA_n$ therefore tends to $1 - m(z)^2/(2\log(1/p))$, a function of $p$ alone, the same for every population and every statistic with positive variance.
It is the dashed curve of Figure~\ref{fig:regime-dial}.
At $p = 1/2$ the share is $1 - 1/(\pi\log 2)\approx 0.54$.
As $z$ grows, two integrations by parts give
\[
  \bar\Phi(z) = \frac{e^{-z^2/2}}{z\sqrt{2\pi}}\,\bigl(1 - z^{-2} + O(z^{-4})\bigr)
\]
Then $m(z) = z + z^{-1} + O(z^{-3})$, so the limit of $\TC(\cA_n)$ is $\log(z\sqrt{2\pi}) - 1 + O(z^{-2})$.
The dependence of a half-space therefore grows as $\log z$, while the exponent grows as $z^2/2$.
For fair signs and their sum, the second limit is equivalent to a limit known in the analysis of Boolean functions.
The level-1 Fourier weight of the indicator of $\cA_n$ is $n\,\Prob_P(\Pemp\in\cA_n)^2$ times the chi-square divergence of $\omega_{\cA_n}$ from $P$, which is twice the relative entropy to first order.
By Proposition~5.25 of \cite{odonnell2014analysis}, that weight tends to $p^2 m(z)^2$.
The half-spaces of the regime map lie below the curve or barely above it and approach it as $n$ grows (Appendix~\ref{sec:eval-regime-map}).
The few far below it have small $n$ and a threshold near the largest value of $\ell$, where the Gaussian approximation to the tail of a bounded statistic is poor.
At finite $n$ the curve is not a bound.
For $\ell = (0,1,2)$ under $P = (0.8,0.15,0.05)$ at $n = 7$, the half-space $\{t : \ip{\ell}{t}\ge 2/7\}$ has probability $0.515$.
Its dependence share is $0.600$, against $0.549$ on the curve.
The probability and the dependence are exact sums over the $36$ types at $n = 7$, by the formulas of Appendix~\ref{sec:eval-regime-map}.
The curve is evaluated at that probability.

\subsection{Curvature}
\label{sec:curvature}

\begin{proposition}[Tangent half-space and the curvature gap]
\label{prop:tangent}
Let $\cA$ be convex with information projection $\Qs$ of full support, $g = \log(d\Qs/dP)$, and $\cH := \{t : \ip{g}{t}\ge\ip{g}{\Qs}\}$.
Then $\cA\subseteq\cH$, $\Qs$ is the information projection of $P$ onto $\cH$, the zero set of $\Delta$ on $\cA$ is $\cA\cap\partial\cH$, and
\begin{equation}
\label{eq:curvature-tax}
  \tau_n(\cA) := \log\Prob_P(\Pemp\in\cH) - \log\Prob_P(\Pemp\in\cA) \ge 0
\end{equation}
with equality if and only if $\cA\cap\cT_n = \cH\cap\cT_n$.
\end{proposition}

The set $\cA\cap\partial\cH$ is the face of $\cA$ exposed by the supporting hyperplane $\partial\cH$.
It is all of $\cA$ for a linear family, the bounding hyperplane for a half-space, a lower-dimensional face for a polytope, and the single point $\Qs$ where the boundary of $\cA$ is strictly convex, as on a relative-entropy ball.
Two events with the same projection and Sanov rate, one flat and one curved, differ in log-probability by $\tau_n$.
Equation~\eqref{eq:three-way} attributes that difference to their transport and dependence terms.

Fix a vector $\gamma$ with $\sum_x\gamma(x) = 0$ and $\ip{g}{\gamma} = 1$, let $V := \{v : \sum_x v(x) = 0,\ \ip{g}{v} = 0\}$, and write $\Pi v$ for the component in $V$ of a vector $v$ with $\sum_x v(x) = 0$, so that $v = \ip{g}{v}\,\gamma + \Pi v$.
Say that $\cA$ is \emph{curved at $\Qs$ with Hessian $M$} if some function $\phi$ on $V$, twice differentiable at $0$ with $\phi(0) = 0$, zero gradient and positive definite Hessian $M$ there, has the property that a distribution $t$ near $\Qs$ lies in $\cA$ exactly when $\ip{g}{t-\Qs}\ge\phi(\Pi(t-\Qs))$.
The Hessian does not depend on $\gamma$: another choice moves the tangential coordinate of a boundary point by a multiple of $\phi(u) = O(|u|^2)$.
Let $\zeta$ be a centered Gaussian vector on $V$ with the law of $\xi\sim N(0,\mathrm{diag}(\Qs)-\Qs{\Qs}^\top)$ conditioned on $\ip{g}{\xi} = 0$, let $\Sigma_\zeta$ be its covariance, and let $q(z) := \tfrac12 z^\top Mz$.
As a quadratic form on $V$, the inverse of $\Sigma_\zeta$ is $v\mapsto\sum_x v(x)^2/\Qs(x)$, the Hessian of the rate function $\dkl{\cdot}{P}$ at $\Qs$.
The form is that Hessian because the second derivative of $s\mapsto s\log(s/P(x))$ is $1/s$, which is $1/\Qs(x)$ at $s=\Qs(x)$.
On $\{v : \sum_x v(x) = 0\}$ it is the inverse of the covariance $\mathrm{diag}(\Qs)-\Qs{\Qs}^\top$ of $\xi$.
Conditioning a Gaussian vector on $\ip{g}{\xi} = 0$ restricts its inverse covariance to $V$.
When $P\notin\cA$ the function $g$ is not constant, and the differences $g(x)-g(y)$ generate a subgroup of $\mathbb{R}$ that is either dense or equal to $\eta\mathbb{Z}$ for some $\eta>0$.
In the second case $g = g(x_0) + \eta s$ with $s$ integer valued, and for all large $n$ the phase $x_n := \lceil n\mu\rceil - n\mu$, with $\mu := \EV_{\Qs}[s]$, places the first layer of types in $\cH$ at height $\eta x_n/n$ above $\partial\cH$ as measured by $\ip{g}{\cdot}$.

\begin{proposition}[The limit of the curvature gap]
\label{prop:curvature-limit}
Let $K\ge 3$, and let $\cA$ be closed, convex and curved at its information projection $\Qs$ with Hessian $M$, where $P\notin\cA$ and $\Qs$ has full support.
\begin{enumerate}
  \item[(i)] If the differences of the values of $g$ generate a dense subgroup of $\mathbb{R}$, then $\tau_n(\cA)\to -\log\EV e^{-q(\zeta)} = \tfrac12\log\det(I + M\Sigma_\zeta)$.
  \item[(ii)] If they generate $\eta\mathbb{Z}$, then $\tau_n(\cA) - F(x_n)\to 0$, where $F(x) := -\log\EV\exp(-\eta\lceil q(\zeta)/\eta - x\rceil)$ is continuous and strictly decreasing on $[0,1)$ and tends to $F(0)-\eta$ as $x\to 1$; so $\tau_n(\cA)$ converges if and only if $\mu$ is an integer.
\end{enumerate}
\end{proposition}

In case (ii) the curvature penalty $q(\zeta)$ is incurred in whole steps of $\eta$, counted from the first layer of types in $\cH$.
When $\mu$ is irrational the phases $x_n$ are dense in $[0,1)$, so the limit points of $\tau_n$ fill $[F(0)-\eta, F(0)]$, an interval of length $\eta$ whatever the curvature.
Asymptotically the probability of the half-space is a factor common to both sets times $e^{-\eta x_n}$.
That factor, the phase-dependent part of the Bahadur--Rao constant for a statistic on an arithmetic progression, oscillates with $n$ unless $\mu$ is an integer \cite{BahadurRao1960}.
The probability of the curved set is the common factor times $\EV e^{-\eta(x_n + \lceil q(\zeta)/\eta - x_n\rceil)}$, so their ratio is $e^{-F(x_n)}$.
The geometry of the set near its projection is the difficult problem behind sharper estimates from \eqref{eq:sharp-sanov} on an open convex set \cite{dinwoodie1992mesures}.
For a sample mean in $\mathbb{R}^d$ with a smooth boundary and a bounded density, sharp asymptotics identify the limiting constant through the second fundamental forms of the boundary and of the level set of the rate function at their point of contact \cite{andriani1997sharp}.
In a similar setting the power of $n$ multiplying the exponential is set by the geometry \cite{iltis1995sharp}.
The Bahadur--Rao constant of a half-space \cite{BahadurRao1960} is the flat case.
For a scalar sequence, strong asymptotics cover both cases, whether or not the values lie on an arithmetic progression \cite{chaganty1993strong}.
For sums of random vectors whose values lie on a regular grid of points, the sharp asymptotics on a general set depend on $n$ through the smallest rate over the grid points in the set \cite{barbe2005sharp}.
Case (i) has the same two ingredients on the set of types: the boundary's Hessian $M$ against the rate function's Hessian $\Sigma_\zeta^{-1}$.
Case (ii) has no counterpart for a density.
On a relative-entropy ball $\{t : \dkl{t}{\pi}\le r\}$ the projection satisfies $g = c - \beta\log(\Qs/\pi)$ for a constant $c$, with $\beta$ its Lagrange multiplier.
Expanding $\dkl{t}{\pi}\le r$ to second order at $\Qs$ then gives $\ip{g}{t-\Qs}\ge\tfrac{\beta}{2}\sum_x (t(x)-\Qs(x))^2/\Qs(x) + o(|t-\Qs|^2)$, so $M = \beta\,\Sigma_\zeta^{-1}$.
Then $I + M\Sigma_\zeta = (1+\beta)I$ on $V$, which has dimension $K-2$, so the limit in case (i) is $\tfrac{K-2}{2}\log(1+\beta)$.
The relative-entropy balls around $\pi = (0.2, 0.3, 0.5)$ in Appendix~\ref{sec:eval-curvature} fall under case (ii): at large $n$ their exact gaps lie close to $F(x_n)$ and keep moving across the range of $F$.

\section{Classical results as special cases}
\label{sec:special-cases}

The finite-sample tools of large deviations are readings of \eqref{eq:identity}: the identity at a particular reference, on a particular geometry, or with its nonnegative terms dropped or bounded.

For a measurable function $h$ on $\cX^n$ with values in $[-\infty,\infty)$ and $0<\EV_{P^{\otimes n}}[e^{h}]<\infty$, the Donsker--Varadhan variational principle reads
\begin{equation}
\label{eq:dv}
  \log\EV_{P^{\otimes n}}\!\left[e^{h}\right] \;=\; \sup_{\rho}\left\{\EV_\rho[h] - \dkl{\rho}{P^{\otimes n}}\right\}
\end{equation}
with the supremum over laws $\rho$ on $\cX^n$.
Write $\rho_h$ for the tilt of $P^{\otimes n}$ by $h$, the law with density $e^{h}/\EV_{P^{\otimes n}}[e^{h}]$ with respect to $P^{\otimes n}$.
At every $\rho$, the left side of \eqref{eq:dv} exceeds the expression in braces by $\dkl{\rho}{\rho_h}$, because $\log(d\rho_h/dP^{\otimes n}) = h-\log\EV_{P^{\otimes n}}[e^{h}]$.
The supremum is therefore attained at $\rho_h$.
With $h$ equal to $0$ on $E_\cA$ and $-\infty$ off it, the tilt $\rho_h$ is $\muA$, and \eqref{eq:dv} is the identity $-\log\Prob_P(\Pemp\in\cA) = \dkl{\muA}{P^{\otimes n}}$ of Lemma~\ref{lem:conditioned}.
The chain rule of Lemma~\ref{lem:chain} splits that identity into the two terms of \eqref{eq:lemma4}.
At any other $\rho$, \eqref{eq:dv} is an inequality with slack $\dkl{\rho}{\muA}$, which is infinite unless $\rho$ is supported on $E_\cA$.
A bound follows from any $\rho$ whose two terms can be controlled.
PAC--Bayes bounds apply the same variational principle over the hypothesis space instead, with $\rho$ a posterior and the reference a data-independent prior \cite{hellstrom2025generalizationpacbayes}.
In standard usage, change of measure names the inequality at a suboptimal $\rho$; Theorem~\ref{thm:identity} is the identity at the optimal $\rho$ of \eqref{eq:dv}, written at a fixed reference $R$ and with two transport terms in place of the supremum.

Sanov's rate replaces $\om$ in \eqref{eq:lemma4} by $\Qs$, a step that needs $\om\in\cA$ (Corollary~\ref{cor:convex-log}).
Because the rate of the most probable type lies within $K\log(n+1)$ of the exponent on every set (Proposition~\ref{prop:bracket}), the value of the classical rate survives the replacement up to the error of approximating the projection by a type.
The mechanism it describes, a sample that deviates by drifting its marginal to the projection, fails on the symmetric events of Section~\ref{sec:taxonomy}, whose marginal does not move.

Chernoff's bound is the identity on the half-space $\cA_\varepsilon = \{t : \ip{\ell}{t}\ge\ip{\ell}{P}+\varepsilon\}$ of a statistic $\ell$, with its nonnegative terms dropped.
When the threshold $\ip{\ell}{P}+\varepsilon$ lies below the largest value of $\ell$, the information projection $\Qs$ is an exponential tilt of $P$ with full support (Lemma~\ref{lem:hoeffding-ratio}(i)), so $P$ and $\Qs$ are mutually absolutely continuous.
The three terms that \eqref{eq:identity} at $R = \Qs$ adds to $n\,\dkl{\Qs}{P}$ are then nonnegative.
Two of them are relative entropies, and the remaining one is $n$ times the defect at the barycenter $\om$.
The half-space is convex and the barycenter is an average of types in it, so the barycenter lies in the half-space, where the defect is nonnegative (Proposition~\ref{prop:tangent}).
Dropping the three terms gives $-\log\Prob_P(\Pemp\in\cA_\varepsilon)\ge n\,\dkl{\Qs}{P}$.
By Lemma~\ref{lem:hoeffding-ratio}(i), $n\,\dkl{\Qs}{P}$ is the exponent of Chernoff's bound.
Corollary~\ref{cor:convex-log} bounds the dropped dependence by $n$ times the error of approximating the projection by a type, plus $K\log(n+1)$.
Hoeffding's bound relaxes the exponent itself, to $2n\varepsilon^2$ when $\ell$ has unit range \cite[Theorem~2]{hoeffding1963}.
Lemma~\ref{lem:hoeffding-ratio}, proved in Appendix~\ref{app:proofs}, compares the two exponents and gives the limit of their ratio as $\varepsilon$ shrinks.

\begin{lemma}[Chernoff's and Hoeffding's exponents on a half-space]
\label{lem:hoeffding-ratio}
Let $\cX$ be finite, let $P$ have full support, and let $\ell$ be a function on $\cX$ with $\mathrm{Var}_P(\ell)>0$.
For $\lambda\in\mathbb{R}$, let $Q_\lambda(x) := P(x)\,e^{\lambda\ell(x)}/\EV_P[e^{\lambda\ell}]$.
For $0<\varepsilon<\max_x\ell(x)-\ip{\ell}{P}$, let $\Qs_\varepsilon$ be the information projection of $P$ onto $\cA_\varepsilon := \{t : \ip{\ell}{t}\ge\ip{\ell}{P}+\varepsilon\}$.
\begin{enumerate}
  \item[(i)] $\Qs_\varepsilon = Q_\lambda$ at the unique $\lambda>0$ with $\ip{\ell}{Q_\lambda} = \ip{\ell}{P}+\varepsilon$, so $\Qs_\varepsilon$ has full support, and $\dkl{\Qs_\varepsilon}{P}$ is the maximum over $s\in\mathbb{R}$ of $s(\ip{\ell}{P}+\varepsilon)-\log\EV_P[e^{s\ell}]$.
  \item[(ii)] $\dkl{\Qs_\varepsilon}{P}/\varepsilon^2\to1/(2\mathrm{Var}_P(\ell))$ as $\varepsilon\downarrow0$.
  \item[(iii)] If the values of $\ell$ lie in an interval of length one, then $\dkl{\Qs_\varepsilon}{P}\ge2\varepsilon^2$ for every such $\varepsilon$.
  Moreover $\mathrm{Var}_P(\ell)\le\tfrac14$, with equality if and only if $\ell$ takes only the two endpoint values of the interval, each with $P$-probability one half.
  \item[(iv)] Under the hypothesis of (iii), $\dkl{\Qs_\varepsilon}{P}/(2\varepsilon^2)\to1/(4\mathrm{Var}_P(\ell))$ as $\varepsilon\downarrow0$, a limit equal to one when $\mathrm{Var}_P(\ell) = \tfrac14$ and larger than one otherwise.
\end{enumerate}
\end{lemma}

The constant $2$ is therefore tight: for a fair two-point loss, $\dkl{\Qs_\varepsilon}{P}\ge c\,\varepsilon^2$ fails at small $\varepsilon$ whenever $c>2$.
For the loss $\ell = (0,\tfrac12,1)$ under $P = (0.5,0.3,0.2)$, the ratio rises toward its limit $1.64$ as $\varepsilon$ decreases through the thresholds of Appendix~\ref{sec:eval-classical}.
When the threshold is the largest value of $\ell$, the half-space is the face on the letters where $\ell$ attains its maximum, so Proposition~\ref{prop:face} applies.

The union bound over $k$ events lowers the best single exponent by $\log k$.
By Theorem~\ref{thm:partition} the pooled dependence is the average within-piece dependence plus $n\,\JS_w$ minus $\Ent(w)$, and $\Ent(w)\le\log k$ is the union bound's own term; the middle term has no counterpart in the classical statement.

The likelihood ratio $L_n = (d\Qs/dP)^{\otimes n}$ of the sample, the change of measure in \eqref{eq:sharp-sanov}, has mean one under $P$, so it is an e-variable for the null hypothesis $P$.
On a closed convex $\cA$ that excludes $P$ it is, under regularity conditions, the growth-optimal e-variable against the alternatives $Q^{\otimes n}$ with $Q\in\cA$ \cite{grunwald2024growthoptimal}.
Its logarithm is $n\,\ip{g}{\Pemp} = n\,\dkl{\Qs}{P} + n\,\Delta(\Pemp)$ by \eqref{eq:defect}.
On a finite alphabet the defect is nonnegative at every type in a convex $\cA$ by the Pythagorean inequality \cite{csisz1975idivergenceb}, so $L_n\ge e^{n\dkl{\Qs}{P}}$ on the event, and Markov's inequality applied to $L_n$ is the Csisz\'ar--Sanov--Chernoff bound.
Equation~\eqref{eq:sharp-sanov} is its exact form, and \eqref{eq:three-way} is its slack.
The bound extends past convexity when the means of a $d$-dimensional statistic under the alternatives surround its mean under $P$, provided the surrounded region is star-shaped about that mean and lies inside the interior of the convex hull of the statistic's values.
The exponent is then at least the smallest rate on the boundary of that region less a minimax regret \cite{grunwald2024growthoptimal}.
For the mean of $n$ draws that regret grows as $\frac{d-1}{2}\log n$, under regularity conditions on the boundary \cite{grunwald2024growthoptimal}.
For $d = 1$ and alternatives that are tilts of $P$, the growth-optimal e-variable is a mixture of the tilts at the two ends of the surrounded interval \cite{grunwald2024growthoptimal}.
The two-sided tails of Figure~\ref{fig:regime-dial} are such events with $d = 1$.
On a uniform population their dependence share is at least $0.70$ (Appendix~\ref{sec:eval-regime-map}).

Gibbs conditioning states that on a convex set any fixed number of conditioned coordinates become independent in the limit, with common law $\Qs$ \cite{csiszar1984Sanov,dembo1996refinements}.
Its refinements for blocks of growing length also require convexity \cite{dembo1996refinements}.
The limit for a fixed number of coordinates is compatible with a whole-sample dependence that grows: on the half-space $\{\ip{\ell}{t}\ge\tfrac12\}$ of the Hoeffding example the total correlation of two coordinates falls toward zero as $n$ grows, while $\TC(\cA)$ rises (Appendix~\ref{sec:eval-classical}).
The rise is at most logarithmic in $n$.
The projection of this half-space has full support (Lemma~\ref{lem:hoeffding-ratio}(i)), so some type $t\in\cA$ lies within $O(1/n)$ of $\Qs$.
For that type, $\dkl{t}{P} - \dkl{\Qs}{P} = \Delta(t) + \dkl{t}{\Qs} = O(1/n)$ (Lemma~\ref{lem:defect-lr}(i)).
Since $\dkl{\Qs_n}{P}\le\dkl{t}{P}$, Corollary~\ref{cor:convex-log} gives $\TC(\cA)\le K\log(n+1) + O(1)$.
The total correlation of two coordinates and $\TC(\cA)$ answer different questions.
Equation~\eqref{eq:lemma4} places $\TC(\cA)$ in the exponent: Corollary~\ref{cor:convex-log} bounds it on a convex set, and on the symmetric unions of Section~\ref{sec:envelope} it is the whole exponent, of order $n$.

\section{Deviation events of learning}
\label{sec:learning}

\subsection{Uniform convergence over a class}

Let $\ell_h : \cX\to[0,1]$ be the loss of a hypothesis $h$ in a finite class $\cF$, with population risk $L(h) = \ip{\ell_h}{P}$ and empirical risk $\ip{\ell_h}{\Pemp}$.
The uniform-convergence event at level $\varepsilon$ is the union of half-spaces
\[
  U_\varepsilon = \bigcup_{h\in\cF}\{t : \ip{\ell_h}{t} - L(h) \ge \varepsilon\}
\]
whose probability the classical treatment bounds by $|\cF|$ times the largest single-hypothesis tail \cite{vapnik1971,hoeffding1963}.
If a group acts on $\cX$ fixing $P$ and permuting the losses $\{\ell_h\}$, the event is invariant.
Take a label-independent classification problem: $\cX = \{1,\ldots,m\}\times\{0,1\}$, $P$ uniform on the inputs with fair labels, and the class of all $2^m$ labelings.
The input permutations and per-input label flips act transitively and fix $P$ and the class, so by Proposition~\ref{prop:symmetry} the marginal rate is zero and the whole exponent is dependence.
The event that the empirical-risk minimizer's population risk exceeds its empirical risk by $\varepsilon$ is invariant in the same way, so its exponent too is all dependence.
A signal in the labels breaks the symmetry, yet over the full class of labelings the share of either event stays close to one, while that of a single hypothesis stays well below one half (Figure~\ref{fig:learning}(a) and Appendix~\ref{sec:eval-classification}).
On the population of iris flowers in the cells of petal-length tercile by species, the three events separate (Figure~\ref{fig:learning}(b)).
Over the class of maps from tercile to species, the dependence share is nearly one for the union, above one half for the minimizer's gap event and below one half for the best single map, a hypothesis fixed before the sample.

\begin{figure}[!t]
\centering
\includegraphics[width=\linewidth]{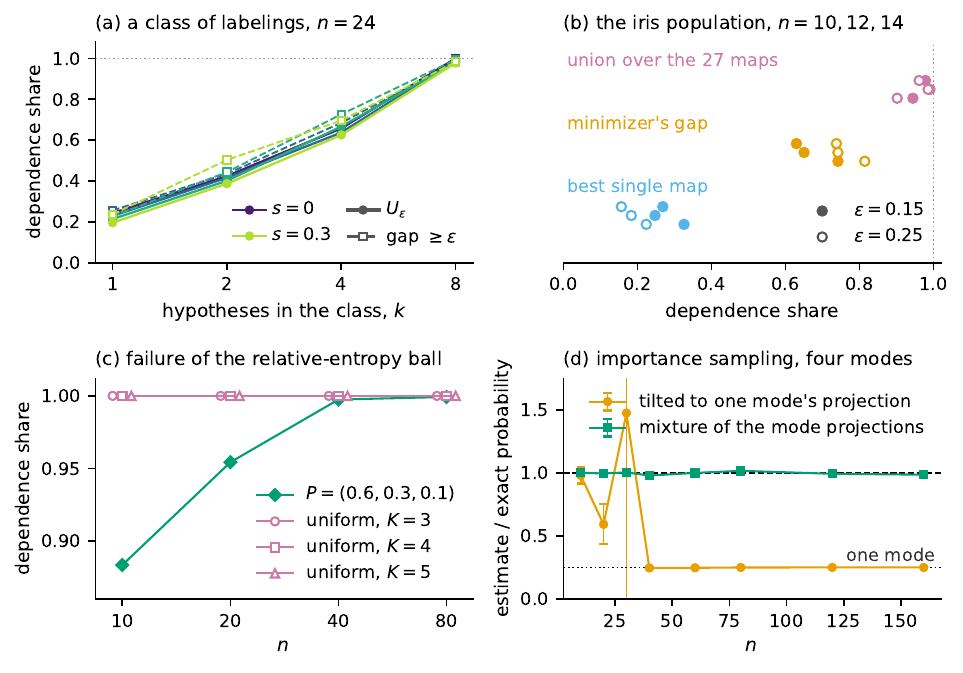}
\caption{\textbf{The whole exponent of a symmetric deviation event of learning is dependence.}
A sampler tilted to one mode of such an event estimates that mode alone.
(a) The dependence share against the size $k$ of a class of labelings of three inputs, at $n = 24$ and $\varepsilon = 0.2$, for the uniform-convergence event $U_\varepsilon$ (solid) and the event that the empirical-risk minimizer's gap is at least $\varepsilon$ (dashed).
The label signal $s$ runs over $0$, $0.1$, $0.2$ and $0.3$ from dark to light.
At $s = 0$ the two curves coincide.
(b) The dependence share of the union over the $27$ maps from tercile to species, the minimizer's gap event and the best single map, on the population of the $150$ iris flowers, at $n = 10, 12, 14$ and $\varepsilon = 0.15$ (filled) and $0.25$ (open).
(c) The dependence share of the failure event $\{\dkl{\Pemp}{P} > r\}$ against $n$ at the exact radius for coverage $0.95$, on uniform populations of three to five letters (open) and on $P = (0.6, 0.3, 0.1)$ (filled).
(d) Importance-sampling estimates of the probability of the symmetric union of four half-spaces $\{t_j\ge 1/2\}$ on four uniform letters, over the exact probability, against $n$, with bars of two reported standard errors.
The dotted line marks one quarter.}
\label{fig:learning}
\end{figure}

The critique that uniform convergence may fail to explain generalization \cite{nagarajan2019uniform} concerns the value of the bound; on a symmetric class the mechanism the bound describes is absent even where its value is right.

\subsection{Interpolation}

The event that a fixed predictor makes no error, so that every sampled point avoids the letters it misclassifies, is the face event $\{\supp\Pemp\subseteq F\}$, and Proposition~\ref{prop:face} gives an independent conditioned sample and the exact classical rate $-n\log P(F)$ at every $n$.
A learner that fits every observed point \cite{zhang2021understanding} sits at the other end when its class is symmetric.
The event that some hypothesis of the class interpolates is a union of such faces.
When a group acting transitively on the alphabet fixes $P$ and permutes the class, that union is invariant, and by Proposition~\ref{prop:symmetry} its whole exponent is dependence.
The unions of faces in Figure~\ref{fig:regime-dial} are these events for classes of predictors that each misclassify one letter.
On a uniform population the whole exponent of the union of all $K$ faces is dependence, while on a skewed one the face of the predictor that misclassifies the least probable letter comes to dominate and the share falls toward zero.
Over a symmetric class, uniform convergence and interpolation, where the classical account of generalization has been questioned, therefore sit at the same end, while the face end belongs to a predictor fixed before the sample.

\subsection{The ambiguity ball}

Distributionally robust optimization replaces the empirical distribution by a divergence ball around it \cite{duchi2021learningc}, and one calibration chooses the radius so that the ball covers the population with prescribed probability \cite{bental2013robust}.
A smaller ball, which may cover the population with small or even zero probability, still secures the statistical guarantee that coverage was meant to provide, through the link to empirical likelihood \cite{lam2019recovering}.
When the radius is instead the exponential rate at which the probability of underestimating an expected cost must decay, the worst case over the relative-entropy ball is the least conservative predictor that achieves that rate \cite{vanparys2021data}.
The coverage event $\{\dkl{\Pemp}{P}\le r\}$ is a curved convex set containing $P$.
The radii in use come from the likelihood-ratio limit $2n\dkl{\Pemp}{P}\to\chi^2_{K-1}$ \cite{vaart1998asymptotic} and from the method-of-types bound $\Prob(\dkl{\Pemp}{P}> r)\le(n+1)^K e^{-nr}$.
Equation~\eqref{eq:lemma4} gives the exact coverage at every $n$ (Appendix~\ref{sec:eval-dro-ball}).
For a uniform $P$ the failure event is invariant under the full symmetric group, so its marginal rate is zero and the exponent that the radius is chosen to control is entirely dependence.

On a skewed population, which no permutation of the letters fixes, the failure event $\{\dkl{\Pemp}{P} > r\}$ at the smallest radius whose coverage reaches $0.95$ still has a dependence share of at least $0.88$.
Its share rises toward one as $n$ grows (Figure~\ref{fig:learning}(c) and Appendix~\ref{sec:eval-dro-ball}).

\subsection{Importance sampling}

Importance sampling under a large-deviations change of measure is efficient in many settings and can fail badly when a rare event has more than one way to occur \cite{glasserman1997counterexamples}.
In terms of types, the failure takes the following form.

\begin{proposition}[Relative variance of the tilted sampler]
\label{prop:sampler}
Let $R$ have full support, let $g = \log(dR/dP)$, and let $X = \ind\{\Pemp\in\cA\}e^{-n\ip{g}{\Pemp}}$ under $\Pemp$ drawn from $R^{\otimes n}$.
Then $\EV[X] = \Prob_P(\Pemp\in\cA)$, and for every $B\subseteq\cA$ with $\Prob_P(\Pemp\in B)>0$ the contribution $X_B = \ind\{\Pemp\in B\}e^{-n\ip{g}{\Pemp}}$ satisfies
\[
  \frac{\EV[X_B^2]}{\EV[X_B]^2} \ge (n+1)^{-3K}\exp\bigl(n\,\dkl{\Qs_{B,n}}{R}\bigr)
\]
with $\Qs_{B,n}$ a most probable type in $B$ under $P$.
\end{proposition}

Taking $R$ the projection onto one piece of a union, the relative variance on any piece whose most probable type is far from $R$ grows exponentially in $n$.
At a feasible number of draws those pieces are never hit, and the estimate settles, with a small reported error, at the probability of the sampler's own piece, one quarter of the truth on a symmetric union of four half-spaces (Figure~\ref{fig:learning}(d) and Appendix~\ref{sec:eval-mixture-limit}).
A proposal that mixes the mode projections gives every mode explicit weight, as the alternatives constructed for such events do \cite{glasserman1997counterexamples}.
Its likelihood ratio on each piece is at most $k$ times the ratio against that piece's projection, and its estimates stay within two percent of the exact value at every sample size enumerated.

\section{Relation to existing results}
\label{sec:related}

\paragraph{Sources of the identity.}
Equation~\eqref{eq:lemma4} appears as equation~(4.13) in the proof of Theorem~1 of \cite{csiszar1984Sanov}.
It follows there from identity~(2.11) of that paper.
The derivation uses no convexity.
Lemma~4 of \cite{balsubramani2020sharp} (version~5) states \eqref{eq:lemma4} for every $\cA$.
Theorem~3 there places the exponent per sample between $\dkl{\om}{P}$ and the cross entropy of $\om$ against $P$ on every set, and states that the comparison with the projection needs convexity.
Theorem~C.4 of \cite{balsubramani2026information} restates \eqref{eq:lemma4} and reads it through the mixed coincidence identity.
The sign of the defect on a convex set rests on the Pythagorean inequality \cite{csisz1975idivergenceb,csiszar2004information}.

\paragraph{The change of measure and the exact difference.}
Section~\ref{sec:setup-tilt} gives the sources of the change of measure \eqref{eq:sharp-sanov}.
Inequality~(2.17) of \cite{csiszar1984Sanov} is \eqref{eq:identity} at the projection onto a convex superset of the event, with the nonnegative defect term dropped.
Theorem~2 of \cite{balsubramani2020sharp} restates it and adds equality on a linear family.
The sources named here give the difference between \eqref{eq:lemma4} and \eqref{eq:sharp-sanov} only as that inequality on a convex set, or as an equality on a linear family, where the defect vanishes.
Equation~\eqref{eq:three-way} states it exactly, with a defect term of either sign.

\paragraph{Microcanonical and canonical ensembles.}
In statistical mechanics the conditioned law is the microcanonical ensemble, which is equivalent to a canonical ensemble where the entropy is concave \cite{lewis1995entropy,touchette2009large}.
At phase coexistence the canonical ensemble is concentrated on several equilibrium states at once \cite{touchette2009large}.
The conditioned law of a multi-mode event, concentrated on its modes, is the analogue of that ensemble.

\paragraph{Exchangeability and mixtures of products.}
The conditioned sample is a finite exchangeable sequence.
The law of any $k$ of its coordinates therefore lies within $2Kk/n$ in total variation of a mixture of independent and identically distributed sequences on an alphabet of $K$ letters, and within $k(k-1)/n$ on any space \cite{diaconis1980finite}.
The whole sample can be far from the product of its own marginal: on a union of separated modes of comparable probability, its total correlation is of order $n$ (Theorem~\ref{thm:partition}).
With the weights $w_j$ and piece marginals $\omega_{\cA_j}$ of that theorem, the relative entropy of $\muA$ from the mixture $\sum_j w_j\,\omega_{\cA_j}^{\otimes n}$ of the pieces' products is at most the average dependence within a piece, $\sum_j w_j\,\TC(\cA_j)$.
This bound follows from the joint convexity of relative entropy \cite{csiszar2004information}, applied to $\muA = \sum_j w_j\,\mu_{\cA_j}$.
The dual total correlation of a law $\mu$ on a finite product space is its entropy less the sum over the coordinates of the conditional entropy of each coordinate given the others.
For points $z$ and $y$, let $\nabla\log\mu(z,y) := \sum_i\,[\log\mu(z_{-i},y_i)-\log\mu(z_{-i},a_i)]$, where $(z_{-i},b)$ is $z$ with its $i$th coordinate set to $b$ and $a_i$ is a reference letter fixed for coordinate $i$.
Fix a partition of the space.
For a point $y$, let $\mu_y$ be $\mu$ conditioned on the piece that contains $y$, and let $\nu_y$ be the product over $i$ of the conditional law under $\mu$ of the $i$th coordinate given the other coordinates of $y$.
Theorem~A of \cite{austin2019structure} assumes that $\mu$ has full support and that, for some $\delta>0$, the partition satisfies $|\nabla\log\mu(z,y)-\nabla\log\mu(z',y)|<\delta n$ at every $y$ whenever $z$ and $z'$ lie in the same piece.
Then the dual total correlation of $\mu$ is less than the entropy of the partition under $\mu$ plus $\delta n$, and so is the average of $\dkl{\mu_y}{\nu_y}$ over $y$ drawn from $\mu$.
When $\delta$ is small and the partition has few enough pieces that its entropy is small compared with $n$, most of the mass of $\mu$ therefore lies on pieces whose conditional laws are close to product laws.
Theorem~A does not apply to the conditioned law $\muA$, which vanishes off $E_\cA$.
For $\muA$, the joint-convexity bound above is the counterpart of Theorem~A.

\section*{Acknowledgments}

We acknowledge the use of large language models in preparing this work.

\bibliography{deviation_event_geometry}

\begin{thebibliography}{38}
\providecommand{\natexlab}[1]{#1}
\providecommand{\url}[1]{\texttt{#1}}
\expandafter\ifx\csname urlstyle\endcsname\relax
  \providecommand{\doi}[1]{doi: #1}\else
  \providecommand{\doi}{doi: \begingroup \urlstyle{rm}\Url}\fi

\bibitem[Andriani and Baldi(1997)]{andriani1997sharp}
Cristina Andriani and Paolo Baldi.
\newblock Sharp estimates of deviations of the sample mean in many dimensions.
\newblock \emph{Annales de l'Institut Henri Poincar{\'e} (B) Probability and
  Statistics}, 33\penalty0 (3):\penalty0 371--385, 1997.

\bibitem[Austin(2019)]{austin2019structure}
Tim Austin.
\newblock The structure of low-complexity {G}ibbs measures on product spaces.
\newblock \emph{The Annals of Probability}, 47\penalty0 (6):\penalty0
  4002--4023, 2019.
\newblock \doi{10.1214/19-AOP1352}.

\bibitem[Bahadur and Rao(1960)]{BahadurRao1960}
Raghu~Raj Bahadur and R.~Ranga Rao.
\newblock On deviations of the sample mean.
\newblock \emph{The Annals of Mathematical Statistics}, 31\penalty0
  (4):\penalty0 1015--1027, 1960.
\newblock \doi{10.1214/aoms/1177705674}.

\bibitem[Balsubramani(2020)]{balsubramani2020sharp}
Akshay Balsubramani.
\newblock Sharp finite-sample concentration of independent variables.
\newblock \emph{arXiv preprint arXiv:2008.13293}, 2020.

\bibitem[Balsubramani(2026)]{balsubramani2026information}
Akshay Balsubramani.
\newblock Information from coincidences.
\newblock \emph{arXiv preprint arXiv:2606.25042}, 2026.

\bibitem[Barbe and Broniatowski(2005)]{barbe2005sharp}
Ph. Barbe and Michel Broniatowski.
\newblock On sharp large deviations for sums of random vectors and
  multidimensional {L}aplace approximation.
\newblock \emph{Theory of Probability \& Its Applications}, 49\penalty0
  (4):\penalty0 561--588, 2005.

\bibitem[Ben-Tal et~al.(2013)Ben-Tal, Den~Hertog, De~Waegenaere, Melenberg, and
  Rennen]{bental2013robust}
Aharon Ben-Tal, Dick Den~Hertog, Anja De~Waegenaere, Bertrand Melenberg, and
  Gijs Rennen.
\newblock Robust solutions of optimization problems affected by uncertain
  probabilities.
\newblock \emph{Management Science}, 59\penalty0 (2):\penalty0 341--357, 2013.

\bibitem[Chaganty and Sethuraman(1993)]{chaganty1993strong}
Narasinga~Rao Chaganty and Jayaram Sethuraman.
\newblock Strong large deviation and local limit theorems.
\newblock \emph{The Annals of Probability}, 21\penalty0 (3):\penalty0
  1671--1690, 1993.

\bibitem[Cover and Thomas(2006)]{cover1999elements}
Thomas~M. Cover and Joy~A. Thomas.
\newblock \emph{Elements of information theory}.
\newblock John Wiley \& Sons, 2nd edition, 2006.
\newblock \doi{10.1002/047174882X}.

\bibitem[Cram{\'e}r(1994)]{cram1938}
Harald Cram{\'e}r.
\newblock Sur un nouveau th{\'e}or{\`e}me-limite de la th{\'e}orie des
  probabilit{\'e}s.
\newblock In \emph{Collected works II}, pages 895--913. Springer, 1994.

\bibitem[Csisz{\'a}r(1975)]{csisz1975idivergenceb}
Imre Csisz{\'a}r.
\newblock {$I$}-divergence geometry of probability distributions and
  minimization problems.
\newblock \emph{The Annals of Probability}, 3\penalty0 (1):\penalty0 146--158,
  1975.
\newblock \doi{10.1214/aop/1176996454}.

\bibitem[Csisz{\'a}r(1984)]{csiszar1984Sanov}
Imre Csisz{\'a}r.
\newblock Sanov property, generalized {$I$}-projection and a conditional limit
  theorem.
\newblock \emph{Annals of Probability}, 12\penalty0 (3):\penalty0 768--793,
  1984.
\newblock \doi{10.1214/aop/1176993227}.

\bibitem[Csisz{\'a}r et~al.(2004)Csisz{\'a}r, Shields,
  et~al.]{csiszar2004information}
Imre Csisz{\'a}r, Paul~C Shields, et~al.
\newblock Information theory and statistics: A tutorial.
\newblock \emph{Foundations and Trends{\textregistered} in Communications and
  Information Theory}, 1\penalty0 (4):\penalty0 417--528, 2004.
\newblock \doi{10.1561/9781933019543}.

\bibitem[Dembo and Zeitouni(1996)]{dembo1996refinements}
Amir Dembo and Ofer Zeitouni.
\newblock Refinements of the {G}ibbs conditioning principle.
\newblock \emph{Probability Theory and Related Fields}, 104:\penalty0 1--14,
  1996.

\bibitem[Dembo and Zeitouni(2010)]{dembo2010large}
Amir Dembo and Ofer Zeitouni.
\newblock \emph{Large Deviations Techniques and Applications}.
\newblock Springer, 2nd, corrected reprint edition, 2010.
\newblock \doi{10.1007/978-3-642-03311-7}.

\bibitem[Diaconis and Freedman(1980)]{diaconis1980finite}
Persi Diaconis and David Freedman.
\newblock Finite exchangeable sequences.
\newblock \emph{The Annals of Probability}, 8\penalty0 (4):\penalty0 745--764,
  1980.
\newblock \doi{10.1214/aop/1176994663}.

\bibitem[Dinwoodie(1992)]{dinwoodie1992mesures}
Ian~H. Dinwoodie.
\newblock Mesures dominantes et th{\'e}or{\`e}me de {S}anov.
\newblock \emph{Annales de l'Institut Henri Poincar{\'e}, Probabilit{\'e}s et
  Statistiques}, 28\penalty0 (3):\penalty0 365--373, 1992.

\bibitem[Duchi and Namkoong(2021)]{duchi2021learningc}
John~C Duchi and Hongseok Namkoong.
\newblock Learning models with uniform performance via distributionally robust
  optimization.
\newblock \emph{The Annals of Statistics}, 49\penalty0 (3):\penalty0
  1378--1406, 2021.

\bibitem[Durrett(2019)]{durrett2019probability}
Rick Durrett.
\newblock \emph{Probability: Theory and Examples}.
\newblock Cambridge University Press, 5th edition, 2019.

\bibitem[Glasserman and Wang(1997)]{glasserman1997counterexamples}
Paul Glasserman and Yashan Wang.
\newblock Counterexamples in importance sampling for large deviations
  probabilities.
\newblock \emph{The Annals of Applied Probability}, 7\penalty0 (3):\penalty0
  731--746, 1997.

\bibitem[Grendar(2006)]{grendar2006conditional}
Marian Grendar.
\newblock Conditional equi-concentration of types.
\newblock \emph{Focus on Probability Theory}, pages 73--89, 2006.

\bibitem[Gr{\"u}nwald et~al.(2024)Gr{\"u}nwald, Hao, and
  Balsubramani]{grunwald2024growthoptimal}
Peter Gr{\"u}nwald, Yunda Hao, and Akshay Balsubramani.
\newblock Growth-optimal {E}-variables and an extension to the multivariate
  {Csisz{\'a}r--Sanov--Chernoff} theorem.
\newblock \emph{arXiv preprint arXiv:2412.17554}, 2024.

\bibitem[Hellstr{\"o}m et~al.(2025)Hellstr{\"o}m, Durisi, Guedj, and
  Raginsky]{hellstrom2025generalizationpacbayes}
Fredrik Hellstr{\"o}m, Giuseppe Durisi, Benjamin Guedj, and Maxim Raginsky.
\newblock Generalization bounds: Perspectives from information theory and
  {PAC}-{Bayes}.
\newblock \emph{Foundations and Trends in Machine Learning}, 18\penalty0
  (1):\penalty0 1--223, 2025.

\bibitem[Hoeffding(1963)]{hoeffding1963}
Wassily Hoeffding.
\newblock Probability inequalities for sums of bounded random variables.
\newblock \emph{Journal of the American Statistical Association}, 58\penalty0
  (301):\penalty0 13--30, 1963.

\bibitem[Iltis(1995)]{iltis1995sharp}
Michael Iltis.
\newblock Sharp asymptotics of large deviations in {$\mathbb{R}^d$}.
\newblock \emph{Journal of Theoretical Probability}, 8\penalty0 (3):\penalty0
  501--522, 1995.

\bibitem[Lam(2019)]{lam2019recovering}
Henry Lam.
\newblock Recovering best statistical guarantees via the empirical
  divergence-based distributionally robust optimization.
\newblock \emph{Operations Research}, 67\penalty0 (4):\penalty0 1090--1105,
  2019.

\bibitem[Lewis et~al.(1995)Lewis, Pfister, and Sullivan]{lewis1995entropy}
John~T Lewis, Charles-Edouard Pfister, and Wayne~G Sullivan.
\newblock Entropy, concentration of probability and conditional limit theorems.
\newblock \emph{Markov Processes and Related Fields}, 1\penalty0 (3):\penalty0
  319--386, 1995.

\bibitem[Nagarajan and Kolter(2019)]{nagarajan2019uniform}
Vaishnavh Nagarajan and J~Zico Kolter.
\newblock Uniform convergence may be unable to explain generalization in deep
  learning.
\newblock \emph{Advances in Neural Information Processing Systems}, 32, 2019.

\bibitem[Ney(1983)]{ney1983dominating}
Peter Ney.
\newblock Dominating points and the asymptotics of large deviations for random
  walk on {$\mathbb{R}^d$}.
\newblock \emph{The Annals of Probability}, 11\penalty0 (1):\penalty0 158--167,
  1983.
\newblock \doi{10.1214/aop/1176993665}.

\bibitem[O'Donnell(2014)]{odonnell2014analysis}
Ryan O'Donnell.
\newblock \emph{Analysis of {B}oolean Functions}.
\newblock Cambridge University Press, 2014.

\bibitem[Sanov(1957)]{Sanov57}
Ivan~Nikolaevich Sanov.
\newblock On the probability of large deviations of random magnitudes.
\newblock \emph{Matematicheskii Sbornik}, 42(84)\penalty0 (1):\penalty0 11--44,
  1957.
\newblock Translation at
  https://repository.lib.ncsu.edu/items/8f909775-ba1b-4874-acc2-362a8221edb0.

\bibitem[Touchette(2009)]{touchette2009large}
Hugo Touchette.
\newblock The large deviation approach to statistical mechanics.
\newblock \emph{Physics Reports}, 478\penalty0 (1--3):\penalty0 1--69, 2009.
\newblock \doi{10.1016/j.physrep.2009.05.002}.

\bibitem[Vaart(1998)]{vaart1998asymptotic}
A.~W. van~der Vaart.
\newblock \emph{Asymptotic Statistics}.
\newblock Cambridge University Press, 1998.
\newblock ISBN 9780511802256.
\newblock \doi{10.1017/cbo9780511802256}.

\bibitem[Van~Parys et~al.(2021)Van~Parys, Esfahani, and Kuhn]{vanparys2021data}
Bart P.~G. Van~Parys, Peyman~Mohajerin Esfahani, and Daniel Kuhn.
\newblock From data to decisions: Distributionally robust optimization is
  optimal.
\newblock \emph{Management Science}, 67\penalty0 (6):\penalty0 3387--3402,
  2021.

\bibitem[Vapnik and Chervonenkis(1971)]{vapnik1971}
V.~N. Vapnik and A.~Y. Chervonenkis.
\newblock On the uniform convergence of relative frequencies of events to their
  probabilities.
\newblock \emph{Theory of Probability and Its Applications}, 16\penalty0
  (2):\penalty0 264--280, 1971.
\newblock \doi{10.1137/1116025}.

\bibitem[Ver~Steeg and Galstyan(2014)]{versteeg2014discovering}
Greg Ver~Steeg and Aram Galstyan.
\newblock Discovering structure in high-dimensional data through correlation
  explanation.
\newblock \emph{Advances in neural information processing systems}, 27, 2014.

\bibitem[Watanabe(1960)]{watanabe1960information}
Satosi Watanabe.
\newblock Information theoretical analysis of multivariate correlation.
\newblock \emph{IBM Journal of Research and Development}, 4\penalty0
  (1):\penalty0 66--82, 1960.
\newblock \doi{10.1147/rd.41.0066}.

\bibitem[Zhang et~al.(2021)Zhang, Bengio, Hardt, Recht, and
  Vinyals]{zhang2021understanding}
Chiyuan Zhang, Samy Bengio, Moritz Hardt, Benjamin Recht, and Oriol Vinyals.
\newblock Understanding deep learning (still) requires rethinking
  generalization.
\newblock \emph{Communications of the ACM}, 64\penalty0 (3):\penalty0 107--115,
  2021.
\newblock \doi{10.1145/3446776}.

\end{thebibliography}

\appendix
\appendixtitle

\section{Proofs}
\label{app:proofs}

Lemma~\ref{lem:conditioned} is the first equality of Lemma~4 of \cite{balsubramani2020sharp}.
Lemma~\ref{lem:chain} is identity~(2.11) of \cite{csiszar1984Sanov} for an exchangeable law.
Lemma~5 of \cite{balsubramani2020sharp} states it for the conditioned law.

\begin{lemma}
\label{lem:conditioned}
$\dkl{\muA}{P^{\otimes n}} = -\log\Prob_P(\Pemp\in\cA)$.
\end{lemma}

\begin{lemma}
\label{lem:chain}
Let $\mu$ be an exchangeable distribution on $\cX^n$ with one-coordinate marginal $\omega$, and let $Q$ be a distribution on $\cX$ with $\omega\ll Q$.
Then $\dkl{\mu}{Q^{\otimes n}} = \dkl{\mu}{\omega^{\otimes n}} + n\,\dkl{\omega}{Q}$.
\end{lemma}

\begin{lemma}
\label{lem:defect-lr}
Let $\Qs\ll P$, and let $g = \log(d\Qs/dP)$, with the value $-\infty$ where the density vanishes.
\begin{enumerate}
  \item[(i)] For every distribution $t$ on $\cX$ with $\dkl{t}{P}$, $\dkl{t}{\Qs}$ and $\dkl{\Qs}{P}$ finite, $g$ is $t$-integrable and $\ip{g}{t - \Qs} = \dkl{t}{P} - \dkl{\Qs}{P} - \dkl{t}{\Qs}$.
  \item[(ii)] If also $P\ll\Qs$, then $P^{\otimes n}\ll(\Qs)^{\otimes n}$ with $dP^{\otimes n}/d(\Qs)^{\otimes n} = e^{-n\ip{g}{\Pemp}}$, which at $z\in\cX^n$ is $\prod_{i=1}^n e^{-g(z_i)}$.
\end{enumerate}
\end{lemma}

\subsection{Proofs for Section~\ref{sec:setup}}

\begin{proof}[Proof of Lemma~\ref{lem:defect-lr}]
(i) Finite $\dkl{t}{P}$ forces $t\ll P$, and finite $\dkl{t}{\Qs}$ forces $t\ll\Qs$.
The chain rule gives $dt/dP = (dt/d\Qs)\,(d\Qs/dP)$ $P$-almost everywhere, and the set where $dt/dP > 0$ has $t$-probability one; on that set both factors are positive.
So $\log(dt/dP) = \log(dt/d\Qs) + g$ holds $t$-almost everywhere with all three terms finite.
For a distribution $\nu$ with $t\ll\nu$, write $\varphi = dt/d\nu$; the negative part of $\log\varphi$ has $\int(\log\varphi)^-\,dt = \int_{\{0<\varphi<1\}}\varphi\log(1/\varphi)\,d\nu\le 1/e$, since $u\log(1/u)\le 1/e$ on $(0,1)$ and $\nu(\{0<\varphi<1\})\le 1$.
Hence $\dkl{t}{\nu} = \int\log\varphi\,dt$ is finite exactly when $\log\varphi$ is $t$-integrable.
At $\nu = P$ and at $\nu = \Qs$ this makes $\log(dt/dP)$ and $\log(dt/d\Qs)$ $t$-integrable, so $g = \log(dt/dP) - \log(dt/d\Qs)$ is $t$-integrable with $\ip{g}{t} = \dkl{t}{P} - \dkl{t}{\Qs}$.
The set where $g = -\infty$ is $\Qs$-null, so $\ip{g}{\Qs} = \int\log(d\Qs/dP)\,d\Qs = \dkl{\Qs}{P}$, finite by hypothesis.
Subtracting gives (i).
(ii) When also $P\ll\Qs$, the chain rule gives $(dP/d\Qs)(d\Qs/dP) = 1$ $P$-almost everywhere, hence $\Qs$-almost everywhere.
So $g$ is finite and $dP/d\Qs = e^{-g}$ $\Qs$-almost everywhere.
For a measurable rectangle $B_1\times\cdots\times B_n$, Tonelli's theorem gives $P^{\otimes n}(B_1\times\cdots\times B_n) = \prod_{i=1}^n\int_{B_i}e^{-g}\,d\Qs = \int_{B_1\times\cdots\times B_n}\prod_{i=1}^n e^{-g(z_i)}\,d(\Qs)^{\otimes n}(z)$.
The rectangles form a $\pi$-system that contains $\cX^n$ and generates the product $\sigma$-algebra, so $P^{\otimes n}$ and the measure with density $\prod_i e^{-g(z_i)}$ against $(\Qs)^{\otimes n}$ agree on every measurable set.
Hence $P^{\otimes n}\ll(\Qs)^{\otimes n}$ with that density, and $\prod_i e^{-g(z_i)} = \exp(-\sum_i g(z_i)) = e^{-n\ip{g}{\Pemp}}$ at $z$, since $\ip{g}{\Pemp} = \frac1n\sum_i g(z_i)$.
\end{proof}

\subsection{Proofs for Section~\ref{sec:results}}

\begin{proof}[Proof of Lemma~\ref{lem:conditioned}]
On $E_\cA$, $\muA(z) = P^{\otimes n}(z)/\Prob_P(\Pemp\in\cA)$, so $\log(\muA/P^{\otimes n})$ equals $-\log\Prob_P(\Pemp\in\cA)$ at every point of the support of $\muA$, and the expectation under $\muA$ is that constant.
\end{proof}

\begin{proof}[Proof of Lemma~\ref{lem:chain}]
Let $G = \{d\omega/dQ > 0\}$.
Each coordinate lies in $G$ with $\mu$-probability one, and $\omega^{\otimes n}$ and $Q^{\otimes n}$ are equivalent on $G^n$, so $\mu\ll\omega^{\otimes n}$ exactly when $\mu\ll Q^{\otimes n}$; otherwise both sides are infinite.
When they hold, $\log(d\mu/dQ^{\otimes n})(z) = \log(d\mu/d\omega^{\otimes n})(z) + \sum_{i=1}^n\log(d\omega/dQ)(z_i)$ on $G^n$.
Both logarithms have $\mu$-integrable negative parts, and taking the expectation under $\mu$, with each coordinate of law $\omega$, gives $\dkl{\mu}{Q^{\otimes n}} = \dkl{\mu}{\omega^{\otimes n}} + n\,\dkl{\omega}{Q}$.
\end{proof}

\begin{proof}[Proof of Theorem~\ref{thm:identity}]
By Lemma~\ref{lem:conditioned}, $-\log\Prob_P(\Pemp\in\cA) = \dkl{\muA}{P^{\otimes n}} = \dkl{\muA}{R^{\otimes n}} + \EV_{\muA}[\log(R^{\otimes n}/P^{\otimes n})]$, and the second term is $\sum_i\EV_{\muA}[g(Z_i)] = n\ip{g}{\om}$ by exchangeability.
By Lemma~\ref{lem:chain} with $Q = R$, which applies by hypothesis, $\dkl{\muA}{R^{\otimes n}} = \TC(\cA) + n\dkl{\om}{R}$.
Finally $\ip{g}{\om} = \ip{g}{R} + \ip{g}{\om - R}$ and $\ip{g}{R} = \dkl{R}{P}$.
\end{proof}

\begin{proof}[Proof of Corollary~\ref{cor:two-references}]
(i) At $R = \om$ both $\ip{g}{\om-R}$ and $\dkl{\om}{R}$ vanish.
(ii) With $g = \log(d\Qs/dP)$, $\ip{g}{\Qs} = \dkl{\Qs}{P}$, and both $\ip{g}{\om-\Qs} = \Delta(\om)$ and $\ip{g}{\Pemp-\Qs} = \Delta(\Pemp)$ hold by \eqref{eq:defect} whether or not the identity of Lemma~\ref{lem:defect-lr}(i) applies.
By Lemma~\ref{lem:defect-lr}(ii) the likelihood ratio of the sample is a function of its empirical measure, $dP^{\otimes n}/d(\Qs)^{\otimes n} = e^{-n\ip{g}{\Pemp}}$, so $\Prob_P(\Pemp\in\cA) = \EV_{\Qs}[\ind\{\Pemp\in\cA\}e^{-n\ip{g}{\Pemp}}] = e^{-n\dkl{\Qs}{P}}\EV_{\Qs}[\ind\{\Pemp\in\cA\}e^{-n\Delta(\Pemp)}]$, which is \eqref{eq:sharp-sanov}.
Substituting \eqref{eq:sharp-sanov} into \eqref{eq:identity} at $R=\Qs$ and canceling $n\dkl{\Qs}{P}$ gives \eqref{eq:three-way}.
\end{proof}

\begin{proof}[Proof of Proposition~\ref{prop:face}]
If $E_\cA = F^n$ then $\muA = P^{\otimes n}(\cdot\mid F^n) = P(\cdot\mid F)^{\otimes n}$, a product, so $\TC(\cA) = 0$.
The marginal is $P(\cdot\mid F)$ and the probability of the event is $P(F)^n$.
On the face, every $Q$ with $Q(F) = 1$ satisfies $\dkl{Q}{P} = \dkl{Q}{P(\cdot\mid F)} - \log P(F)$, since $dP(\cdot\mid F)/dP = \ind_F/P(F)$; the right side is minimized exactly at $Q = P(\cdot\mid F)$, so $\Qs = P(\cdot\mid F)$ and $\dkl{\Qs}{P} = -\log P(F)$.
Then $g = -\log P(F)$ on $F$, so $\Delta(Q) = \ip{g}{Q - \Qs} = 0$ at every $Q$ of the face.
Conversely, suppose $\TC(\cA) = 0$, so $\muA = \om^{\otimes n}$, and write $h = d\om/dP$ and $p = \Prob_P(\Pemp\in\cA)$.
Then $d\muA/dP^{\otimes n} = \ind_{E_\cA}/p$ takes the two values $0$ and $1/p$, so $\prod_i h(z_i)\in\{0,1/p\}$ for $P^{\otimes n}$-almost every $z$.
Put $F = \{h > 0\}$, which has $P(F) > 0$.
On $F^n$, which has positive measure, $\sum_i\log h(z_i) = -\log p$ almost everywhere; freezing all coordinates but one and applying Fubini makes $\log h$ constant $P$-almost everywhere on $F$.
Hence $h = \ind_F/P(F)$, so $\om = P(\cdot\mid F)$ and $P(F) = p^{1/n}$, and $\ind_{E_\cA} = \ind_{F^n}$ almost everywhere.
On a discrete alphabet with $P$ of full support every nonempty set has positive probability, so the exception set is empty and $\cA\cap\cT_n = \{t\in\cT_n : \supp t\subseteq F\}$.
\end{proof}

\begin{proof}[Proof of Proposition~\ref{prop:symmetry}]
For $\sigma\in G$ acting diagonally on $\cX^n$, $\hat P_n(\sigma z) = \sigma\cdot\hat P_n(z)$, so $E_\cA$ is invariant, and $P^{\otimes n}$ is invariant; hence $\muA$ is invariant and so is its one-coordinate marginal $\om$.
Write $h = d\om/dP$.
For each $\sigma$ the function $h\circ\sigma^{-1}$ is a version of $d(\sigma_*\om)/d(\sigma_*P) = d\om/dP$, so, applying this to $\sigma^{-1}$, $h\circ\sigma = h$ holds $P$-almost everywhere.
Averaging over the Haar probability measure of the compact group $G$, the function $\tilde h(x) = \int_G h(\sigma x)\,d\sigma$ is invariant at every point, hence $\cI$-measurable, and exchanging the two integrations shows it equals $h$ $P$-almost everywhere.
For $A\in\cI$ we have $\int_A\tilde h\,dP = \om(A)$, so $\tilde h$ is also a version of $d\om|_\cI/dP|_\cI$, and $\dkl{\om}{P} = \int\log\tilde h\,d\om = \dkl{\om|_\cI}{P|_\cI}$.
If $\cI$ is $P$-trivial then $\tilde h$ is $P$-almost everywhere constant and integrates to one, so $\om = P$.
Then \eqref{eq:lemma4} reduces to $-\log\Prob_P(\Pemp\in\cA) = \TC(\cA)$.
\end{proof}

\begin{proof}[Proof of Theorem~\ref{thm:partition}]
Let $J$ be the index of the piece containing $\Pemp$, a function of the sample.
Under $\muA$, $\Prob(J = j) = w_j$ and the conditional law of the sample given $J=j$ is $\mu_{\cA_j}$, so $\om = \sum_j w_j\omega_{\cA_j}$ by marginalizing.
For a mixture $\mu = \sum_j w_j\mu_j$ and any $\nu$, the compensation identity $\sum_j w_j\dkl{\mu_j}{\nu} = \dkl{\mu}{\nu} + \sum_j w_j\dkl{\mu_j}{\mu}$ holds on any measurable space.
Apply it at $\nu = \om^{\otimes n}$ with $\mu = \muA$ and $\mu_j = \mu_{\cA_j}$.
Because the pieces are disjoint, $d\mu_{\cA_j}/d\muA = \ind\{J=j\}/w_j$, so $\dkl{\mu_{\cA_j}}{\muA} = \log(1/w_j)$ and the last sum is $\Ent(w)$.
Each $\mu_{\cA_j}$ is exchangeable with marginal $\omega_{\cA_j}\ll\om$, since $\om\ge w_j\omega_{\cA_j}$, so Lemma~\ref{lem:chain} at $Q = \om$ gives $\dkl{\mu_{\cA_j}}{\om^{\otimes n}} = \TC(\cA_j) + n\,\dkl{\omega_{\cA_j}}{\om}$.
Collecting the terms and using $\sum_j w_j\dkl{\omega_{\cA_j}}{\om} = \JS_w$ gives \eqref{eq:partition}.
Every term is finite, since $\TC(\cA_j)$ is at most the exponent of its piece and $\Ent(w)\le\log k$.
\end{proof}

\subsection{Proofs for Section~\ref{sec:finite}}

\begin{proof}[Proof of Proposition~\ref{prop:bracket}]
For a type $t$, $(n+1)^{-K}e^{-n\dkl{t}{P}}\le\Prob_P(\Pemp = t)\le e^{-n\dkl{t}{P}}$, and $|\cT_n|\le(n+1)^K$ \cite{cover1999elements}.
Hence $\Prob_P(\Pemp\in\cA)\ge\Prob_P(\Pemp=\Qs_n)\ge(n+1)^{-K}e^{-n\dkl{\Qs_n}{P}}$ and $\Prob_P(\Pemp\in\cA)\le|\cA\cap\cT_n|\,e^{-n\dkl{\Qs_n}{P}}\le(n+1)^Ke^{-n\dkl{\Qs_n}{P}}$, so $|-\log\Prob_P(\Pemp\in\cA) - n\dkl{\Qs_n}{P}|\le K\log(n+1)$.
Subtracting \eqref{eq:lemma4} gives \eqref{eq:bracket}.
\end{proof}

\begin{proof}[Proof of Corollary~\ref{cor:convex-log}]
$\om = \EV_{\muA}[\Pemp]$ is a convex combination of points of $\cA$, so $\om\in\cA$ and $\dkl{\om}{P}\ge\dkl{\Qs}{P}$; substitute into the upper half of \eqref{eq:bracket}.
\end{proof}

\begin{proof}[Proof of Corollary~\ref{cor:mixture-limit}]
Apply Theorem~\ref{thm:identity} to the piece $\cA_j$ at $R = \Qs_j$: $-\log\Prob_P(\Pemp\in\cA_j) = n\dkl{\Qs_j}{P} + n\Delta_j(\omega_{\cA_j}) + n\dkl{\omega_{\cA_j}}{\Qs_j} + \TC(\cA_j)$, where $\Delta_j(\omega_{\cA_j})\ge 0$ by the Pythagorean inequality since $\omega_{\cA_j}\in\cA_j$ by convexity.
The proof of Proposition~\ref{prop:bracket} gives $-\log\Prob_P(\Pemp\in\cA_j)\le n\dkl{\Qs_{j,n}}{P} + K\log(n+1)$, so $n\dkl{\omega_{\cA_j}}{\Qs_j} + \TC(\cA_j)\le n[\dkl{\Qs_{j,n}}{P} - \dkl{\Qs_j}{P}] + K\log(n+1)\le c + K\log(n+1)$, which proves the two bounds.
By Pinsker's inequality $\omega_{\cA_j}\to\Qs_j$.
On a finite alphabet $\JS_w(\omega_1,\ldots,\omega_k) = \Ent(\sum_j w_j\omega_j) - \sum_j w_j\Ent(\omega_j)$, which is continuous in the weights and the arguments jointly on a compact set, hence uniformly continuous.
Therefore $\JS_{w(n)}(\omega_{\cA_1},\ldots,\omega_{\cA_k}) - \JS_{w(n)}(\Qs_1,\ldots,\Qs_k)\to 0$ uniformly in $w(n)$.
Dividing \eqref{eq:partition} by $n$, the within-piece term is at most $(c + K\log(n+1))/n$ and $\Ent(w)/n\le\log k/n$.
\end{proof}

\begin{proof}[Proof of Proposition~\ref{prop:halfspace-gauss}]
Write $S_k := \sum_{i\le k}(\ell(Z_i) - \mu)$ and $s_n := z\sigma\sqrt n$, so that $\{\Pemp\in\cA_n\} = \{S_n\ge s_n\}$.
The central limit theorem gives $\Prob_P(S_n\ge s_n)\to\bar\Phi(z)$, the first limit in \eqref{eq:halfspace-gauss}.
By exchangeability $\omega_{\cA_n}(x) = P(x)\,\Prob_P(S_{n-1}\ge s_n - \ell(x) + \mu)/\Prob_P(S_n\ge s_n)$, and averaging the numerator over $x\sim P$ recovers the denominator, so
\[
  \frac{\omega_{\cA_n}(x)}{P(x)} - 1 = \frac{\sum_y P(y)\bigl[\Prob_P(S_{n-1}\ge s_n-\ell(x)+\mu) - \Prob_P(S_{n-1}\ge s_n-\ell(y)+\mu)\bigr]}{\Prob_P(S_n\ge s_n)}
\]
Each bracket is plus or minus the probability that $S_{n-1}$ lies in a half-open interval of length $|\ell(x)-\ell(y)|$ at distance $O(1)$ from $s_n$.
If $\ell(Z_1)$ is non-lattice, that probability is $|\ell(x)-\ell(y)|\,e^{-z^2/2}/(\sigma\sqrt{2\pi n}) + o(n^{-1/2})$ by the local limit theorem for non-lattice sums \cite[Theorem~3.5.4]{durrett2019probability}, and by the same theorem each endpoint has probability $o(n^{-1/2})$.
If it is lattice with span $h$, every difference $\ell(x)-\ell(y)$ is a multiple of $h$, so the interval contains exactly $|\ell(x)-\ell(y)|/h$ of the possible values of $S_{n-1}$.
By the lattice local limit theorem \cite[Theorem~3.5.3]{durrett2019probability} each of these values has probability $h\,e^{-z^2/2}/(\sigma\sqrt{2\pi n}) + o(n^{-1/2})$.
Since $\sum_y P(y)(\ell(x)-\ell(y)) = \ell(x)-\mu$ and the denominator tends to $\bar\Phi(z) > 0$, in both cases, uniformly over the finite alphabet,
\[
  \frac{\omega_{\cA_n}(x)}{P(x)} - 1 = \frac{m(z)\,(\ell(x) - \mu)}{\sigma\sqrt n} + o(n^{-1/2})
\]
Since $\dkl{\omega}{P} = \frac12\sum_x P(x)\bigl(\omega(x)/P(x) - 1\bigr)^2 + O\bigl(\max_x|\omega(x)/P(x) - 1|^3\bigr)$, the second limit follows:
\[
  n\,\dkl{\omega_{\cA_n}}{P} = \frac{m(z)^2}{2\sigma^2}\sum_x P(x)(\ell(x)-\mu)^2 + o(1) = \frac{m(z)^2}{2} + o(1)
\]
The third follows from the first two and \eqref{eq:lemma4}.
\end{proof}

\begin{proof}[Proof of Proposition~\ref{prop:tangent}]
The Pythagorean inequality gives $\Delta(t) = \ip{g}{t-\Qs}\ge 0$ for $t\in\cA$, so $\cA\subseteq\cH$, and $\Delta(t) = 0$ on $\cA$ exactly when $\ip{g}{t} = \ip{g}{\Qs}$, which is $\cA\cap\partial\cH$.
For $t\in\cH$, $\dkl{t}{P} = \dkl{\Qs}{P} + \dkl{t}{\Qs} + \Delta(t)\ge\dkl{\Qs}{P}$ with equality only at $t = \Qs$, and $\Qs\in\cH$, so $\Qs$ is the projection onto $\cH$.
Since $\cA\subseteq\cH$, $\Prob_P(\Pemp\in\cH)\ge\Prob_P(\Pemp\in\cA)$, with equality if and only if no type of positive probability lies in $\cH\setminus\cA$; every type has positive probability under a full-support $P$.
\end{proof}

\begin{proof}[Proof of Proposition~\ref{prop:curvature-limit}]
Let $Z_1,\ldots,Z_n$ be independent with law $\Qs$, and set $Y_i := g(Z_i)-\ip{g}{\Qs}$, $W_n := \sum_i Y_i = n\ip{g}{\Pemp-\Qs}$ and $\zeta_n := \sqrt n\,\Pi(\Pemp-\Qs)$.
The $Y_i$ are bounded and centered, with variance $\sigma^2>0$ because $P\notin\cA$ makes $g$ nonconstant.
The likelihood ratio of Lemma~\ref{lem:defect-lr}(ii) gives, for $S\in\{\cA,\cH\}$,
\begin{equation}
\label{eq:gap-tilted}
  e^{n\dkl{\Qs}{P}}\,\Prob_P(\Pemp\in S) = \EV\bigl[\ind\{\Pemp\in S\}\,e^{-W_n}\bigr]
\end{equation}
and $\Pemp\in\cH$ exactly when $W_n\ge 0$.

Let $\lambda>0$ be the smallest eigenvalue of $M$, so that $\phi(u) = \tfrac12u^\top Mu + o(|u|^2)$ and $\phi(u)\ge\lambda|u|^2/4$ for small $u$.
A point $t\neq\Qs$ of $\cA\cap\partial\cH$ would put the segment from $\Qs$ to $t$ in $\cA\cap\partial\cH$, where $\ip{g}{\cdot-\Qs}$ vanishes while $\phi$ is positive near $0$ off the origin; so $\cA\cap\partial\cH = \{\Qs\}$.
The compact sets $\{t\in\cA : \ip{g}{t-\Qs}\le\varepsilon\}$ therefore shrink to $\Qs$ as $\varepsilon\downarrow 0$, and for each $L>0$ and all large $n$
\begin{equation}
\label{eq:gap-local}
  \{\Pemp\in\cA,\ W_n\le L\} = \{q_n(\zeta_n)\le W_n\le L,\ |\zeta_n|^2\le 4L/\lambda\},
  \qquad q_n(z) := n\,\phi(z/\sqrt n)
\end{equation}
where $q_n\to q$ uniformly on bounded sets.

The vectors $X_i := (\Pi(\delta_{Z_i}-\Qs), Y_i)\in V\times\mathbb{R}$ are independent, bounded and centered, with $\zeta_n = n^{-1/2}\sum_i\Pi(\delta_{Z_i}-\Qs)$ and $W_n = \sum_i Y_i$.
Their covariance $\Gamma$ is nonsingular, because $v\mapsto(\Pi v,\ip{g}{v})$ maps $\{v : \sum_x v(x) = 0\}$ isomorphically onto $V\times\mathbb{R}$ and $\mathrm{diag}(\Qs)-\Qs{\Qs}^\top$ is positive definite on that hyperplane when $\Qs$ has full support.
Conditioning a $N(0,\Gamma)$ vector on its last coordinate being zero leaves a vector with the law of $\zeta$.
Let $\psi(a,u) := \EV\exp(i\ip{a}{\Pi(\delta_{Z_1}-\Qs)} + iuY_1)$, so that $\EV\exp(i\ip{a}{\zeta_n}+iuW_n) = \psi(a/\sqrt n,u)^n$.
There are $c,\kappa>0$ with $|\psi(a,u)|\le\exp(-c(|a|^2+u^2))$ when $|a|^2+u^2\le\kappa^2$, and by the central limit theorem
\[
  \psi(a/\sqrt n, v/\sqrt n)^n\to\exp\bigl(-\tfrac12(a,v)\Gamma(a,v)^\top\bigr),
  \qquad
  \int_{\mathbb{R}}\exp\bigl(-\tfrac12(a,v)\Gamma(a,v)^\top\bigr)\,dv = \frac{\sqrt{2\pi}}{\sigma}\,\exp\bigl(-\tfrac12a^\top\Sigma_\zeta a\bigr)
\]
for fixed $a\in V$ and $v\in\mathbb{R}$.

In case (ii), $W_n$ takes values in $\Lambda_n := \eta(\mathbb{Z}-n\mu)$, whose least nonnegative element is $\eta x_n$, and $|\EV e^{iuY_1}|<1$ for $0<|u|\le\pi/\eta$, since equality would force $u(g(x)-g(y))\in 2\pi\mathbb{Z}$ for all $x,y$ and hence $u\eta\in2\pi\mathbb{Z}$.
Fourier inversion on $\eta\mathbb{Z}$ gives, for $w\in\Lambda_n$,
\[
  \sqrt n\,\EV\bigl[e^{i\ip{a}{\zeta_n}}\,\ind\{W_n = w\}\bigr] = \frac{\eta}{2\pi}\int_{|v|\le\pi\sqrt n/\eta} e^{-ivw/\sqrt n}\,\psi(a/\sqrt n, v/\sqrt n)^n\,dv
\]
The integrand is at most $e^{-cv^2}$ for $|v|\le\kappa\sqrt n/2$ and at most $\theta^n$ beyond, for some $\theta<1$ and all large $n$, by continuity of $\psi$.
Dominated convergence shows that the left side tends to $\eta(\sigma\sqrt{2\pi})^{-1}\exp(-\tfrac12a^\top\Sigma_\zeta a)$ uniformly in $w\in\Lambda_n\cap[-L,L]$, and that $\sqrt n\,\Prob(W_n = w)\le C$ for all $n$ and $w$.
By the continuity theorem the measures $B\mapsto\sqrt n\,\Prob(\zeta_n\in B,\ W_n = w_n)$ converge weakly to $\eta(\sigma\sqrt{2\pi})^{-1}$ times the law of $\zeta$ along every sequence $w_n\in\Lambda_n\cap[-L,L]$.
Fix $\varepsilon\in(0,L)$.
For large $n$, $|q_n-q|\le\varepsilon$ on $\{|z|^2\le 4L/\lambda\}$, and $q(z)\ge\lambda|z|^2/2$, so by \eqref{eq:gap-local}, on the event $\{W_n = w\}$ with $w\in[0,L]$,
\[
  \{q(\zeta_n)\le w-\varepsilon\}\subseteq\{\Pemp\in\cA\}\subseteq\{q(\zeta_n)\le w+\varepsilon\}
\]
The law of $q(\zeta)$, a positive combination of independent $\chi^2_1$ variables, has a continuous distribution function, so the sets $\{z : q(z)\le c\}$ have boundaries null for the law of $\zeta$.
Along every sequence $(w_n, c_n)$ with $c_n\to c$, weak convergence then gives $\sqrt n\,\Prob(q(\zeta_n)\le c_n,\ W_n = w_n)\to\eta(\sigma\sqrt{2\pi})^{-1}\Prob(q(\zeta)\le c)$, so the convergence is uniform in $w$ and $c$, and letting $\varepsilon\downarrow0$ through the uniform continuity of that distribution function gives
\[
  \sqrt n\,\Prob(\Pemp\in\cA,\ W_n = w) = \frac{\eta}{\sigma\sqrt{2\pi}}\,\Prob(q(\zeta)\le w) + o(1)
\]
uniformly in $w\in\Lambda_n\cap[0,L]$.
The terms of \eqref{eq:gap-tilted} with $W_n>L$ contribute at most $Ce^{-L}/((1-e^{-\eta})\sqrt n)$.
Writing $w = \eta(x_n+j)$ with $j\ge 0$ and letting $n\to\infty$ and then $L\to\infty$,
\begin{align*}
  \sqrt n\,e^{n\dkl{\Qs}{P}}\Prob_P(\Pemp\in\cH) &= \frac{\eta}{\sigma\sqrt{2\pi}}\sum_{j\ge0}e^{-\eta(x_n+j)} + o(1) = \frac{\eta\,e^{-\eta x_n}}{\sigma\sqrt{2\pi}\,(1-e^{-\eta})} + o(1)\\
  \sqrt n\,e^{n\dkl{\Qs}{P}}\Prob_P(\Pemp\in\cA) &= \frac{\eta}{\sigma\sqrt{2\pi}}\,\EV\!\!\sum_{j\ge\lceil q(\zeta)/\eta-x_n\rceil}\!\!e^{-\eta(x_n+j)} + o(1) = \frac{\eta\,\EV e^{-\eta(x_n+\lceil q(\zeta)/\eta-x_n\rceil)}}{\sigma\sqrt{2\pi}\,(1-e^{-\eta})} + o(1)
\end{align*}
where $\lceil q(\zeta)/\eta-x_n\rceil\ge0$ because $q(\zeta)\ge0$ and $x_n<1$.
Both leading terms are bounded below uniformly in $x_n$, the second because $x+\lceil q/\eta-x\rceil\le q/\eta+1$, so taking logarithms gives $\tau_n-F(x_n)\to0$.

In case (i), $|\EV e^{iuY_1}|<1$ for every $u\neq0$, since the values of $Y_1$ differ by elements of a dense subgroup.
Let $h_b(w) := (\sin(w/b)/(w/b))^2$, whose Fourier transform vanishes outside $[-2/b,2/b]$ and which is bounded below on $[-L,L]$ once $\pi b>L$.
The function $w\mapsto h_b(w-w_0)e^{iuw}$ again has a Fourier transform of bounded support, and expanding it through that transform, the same argument gives
\[
  \sqrt n\,\EV\bigl[h_b(W_n-w_0)\,e^{iuW_n+i\ip{a}{\zeta_n}}\bigr] \longrightarrow \frac{1}{\sigma\sqrt{2\pi}}\int h_b(w-w_0)\,e^{iuw}\,dw\;\exp(-\tfrac12a^\top\Sigma_\zeta a)
\]
for all $w_0$, $u$ and $a$, together with a bound on $\sqrt n\,\EV h_b(W_n-w_0)$, and hence on $\sqrt n\,\Prob(W_n\in[w_0,w_0+1])$, uniform in $w_0$ and $n$.
By the continuity theorem the finite measures $h_b(w)\sqrt n\,\Prob(W_n\in dw,\ \zeta_n\in dz)$ converge weakly to $(\sigma\sqrt{2\pi})^{-1}h_b(w)\,dw$ times the law of $\zeta$.
Dividing by $h_b$ gives $\sqrt n\,\EV f(W_n,\zeta_n)\to(\sigma\sqrt{2\pi})^{-1}\int\EV f(w,\zeta)\,dw$ for every bounded $f$ that vanishes when $w\notin[-L,L]$ and is continuous off a set null for $dw$ times the law of $\zeta$.
With $f = e^{-w}\ind\{0\le w\le L\}$, and through \eqref{eq:gap-local} with $f = e^{-w}\ind\{q(z)\pm\varepsilon\le w\le L\}$, and with the part $W_n>L$ at most $Ce^{-L}/\sqrt n$, this gives $\sqrt n\,e^{n\dkl{\Qs}{P}}\Prob_P(\Pemp\in\cH)\to(\sigma\sqrt{2\pi})^{-1}$ and
\[
  \sqrt n\,e^{n\dkl{\Qs}{P}}\Prob_P(\Pemp\in\cA) \longrightarrow \frac{1}{\sigma\sqrt{2\pi}}\int_0^\infty e^{-w}\,\Prob(q(\zeta)\le w)\,dw = \frac{\EV e^{-q(\zeta)}}{\sigma\sqrt{2\pi}}
\]
So $\tau_n\to-\log\EV e^{-q(\zeta)}$, and the Gaussian integral gives $\EV e^{-q(\zeta)} = \det(I+M\Sigma_\zeta)^{-1/2}$.

The function $F$ is continuous because $q(\zeta)$ has no atoms.
For $0\le x<x'<1$, $\lceil q(\zeta)/\eta-x'\rceil\le\lceil q(\zeta)/\eta-x\rceil$, with strict inequality when $q(\zeta)/\eta\in(x+j,x'+j]$ for an integer $j\ge0$, an event of positive probability because $q(\zeta)$ has a positive density on $(0,\infty)$; so $F$ is strictly decreasing.
As $x\to1$, $\lceil q(\zeta)/\eta-x\rceil\to\lceil q(\zeta)/\eta\rceil-1$ almost surely, so $F(x)\to F(0)-\eta$.
If $\mu$ is an integer then $x_n = 0$ for every $n$ and $\tau_n\to F(0)$.
Otherwise $x_{n+1}-x_n\equiv-\mu\pmod 1$ keeps consecutive phases a fixed positive distance apart on the circle formed by joining the ends of $[0,1)$.
Convergence of $F(x_n)$ would force $x_n$ to converge on that circle, since $F$ is continuous and strictly decreasing with $F(1^-) = F(0)-\eta<F(0)$, so $\tau_n$ has no limit.
\end{proof}

\subsection{Proof for Section~\ref{sec:special-cases}}

\begin{proof}[Proof of Lemma~\ref{lem:hoeffding-ratio}]
Write $\Psi(\lambda) := \log\EV_P[e^{\lambda\ell}]$, so that $Q_\lambda(x) = P(x)\,e^{\lambda\ell(x)-\Psi(\lambda)}$, which is positive at every letter because $P$ is.
On a finite alphabet $\Psi$ is infinitely differentiable, with $\Psi'(\lambda) = \ip{\ell}{Q_\lambda}$ and $\Psi''(\lambda) = \mathrm{Var}_{Q_\lambda}(\ell)$.
That variance is positive because $Q_\lambda$ has full support and $\ell$ is not constant, so $\Psi'$ is strictly increasing, with $\Psi'(0) = \ip{\ell}{P}$.
As $\lambda\to\infty$ the law $Q_\lambda$ concentrates on the letters where $\ell$ is largest, so $\Psi'(\lambda)\to\max_x\ell(x)$.
Hence for $0<\varepsilon<\max_x\ell(x)-\ip{\ell}{P}$ there is a unique $\lambda_\varepsilon>0$ with $\Psi'(\lambda_\varepsilon) = \ip{\ell}{P}+\varepsilon$, and $\lambda_\varepsilon\to0$ as $\varepsilon\downarrow0$, since the inverse of $\Psi'$ is continuous.

(i) Fix $\varepsilon$ and write $\lambda = \lambda_\varepsilon$.
Since $\log(Q_\lambda/P) = \lambda\ell-\Psi(\lambda)$, the expectation under $Q_\lambda$ gives $\dkl{Q_\lambda}{P} = \lambda\Psi'(\lambda)-\Psi(\lambda) = \lambda(\ip{\ell}{P}+\varepsilon)-\Psi(\lambda)$.
For every distribution $t$ on $\cX$ the same logarithm gives $\dkl{t}{P} = \dkl{t}{Q_\lambda}+\lambda\ip{\ell}{t}-\Psi(\lambda)$.
For $t\in\cA_\varepsilon$ this is at least $\dkl{t}{Q_\lambda}+\dkl{Q_\lambda}{P}$, since $\lambda>0$ and $\ip{\ell}{t}\ge\ip{\ell}{P}+\varepsilon$.
So $\dkl{t}{P}\ge\dkl{Q_\lambda}{P}$ on $\cA_\varepsilon$, with equality only at $t = Q_\lambda$, which lies in $\cA_\varepsilon$; hence $\Qs_\varepsilon = Q_{\lambda_\varepsilon}$.
The function $s\mapsto s(\ip{\ell}{P}+\varepsilon)-\Psi(s)$ is concave, and its derivative $\ip{\ell}{P}+\varepsilon-\Psi'(s)$ vanishes at $s = \lambda$, so its maximum over $\mathbb{R}$ is its value $\dkl{\Qs_\varepsilon}{P}$ at $\lambda$.

(ii) The function $\lambda\mapsto\lambda\Psi'(\lambda)-\Psi(\lambda)$ vanishes at $0$ and has derivative $\lambda\Psi''(\lambda)$.
So $\dkl{\Qs_\varepsilon}{P} = \int_0^{\lambda_\varepsilon}s\,\Psi''(s)\,ds$, while $\varepsilon = \Psi'(\lambda_\varepsilon)-\Psi'(0) = \int_0^{\lambda_\varepsilon}\Psi''(s)\,ds$.
The function $\Psi''$ is differentiable with $\Psi''(0) = \mathrm{Var}_P(\ell)$, so $\Psi''(s) = \mathrm{Var}_P(\ell)+O(s)$ as $s\to0$.
Hence $\varepsilon = \lambda_\varepsilon\mathrm{Var}_P(\ell)+O(\lambda_\varepsilon^2)$ and $\dkl{\Qs_\varepsilon}{P} = \tfrac12\lambda_\varepsilon^2\mathrm{Var}_P(\ell)+O(\lambda_\varepsilon^3)$ as $\varepsilon\downarrow0$.
Dividing the second expansion by the square of the first gives $\dkl{\Qs_\varepsilon}{P}/\varepsilon^2\to1/(2\mathrm{Var}_P(\ell))$.

(iii) Let the values of $\ell$ lie in $[a,a+1]$.
Hoeffding's lemma \cite[inequality~(4.16)]{hoeffding1963} gives $\Psi(s)\le s\ip{\ell}{P}+s^2/8$ for every $s\in\mathbb{R}$.
By (i), $\dkl{\Qs_\varepsilon}{P}\ge\max_s\{s(\ip{\ell}{P}+\varepsilon)-s\ip{\ell}{P}-s^2/8\} = \max_s\{s\varepsilon-s^2/8\} = 2\varepsilon^2$.
The variance is $\mathrm{Var}_P(\ell) = \EV_P[(\ell-a-\tfrac12)^2]-(\ip{\ell}{P}-a-\tfrac12)^2\le\EV_P[(\ell-a-\tfrac12)^2]\le\tfrac14$, since $|\ell-a-\tfrac12|\le\tfrac12$.
Equality throughout forces $\ip{\ell}{P} = a+\tfrac12$ and, because $P$ has full support, $|\ell-a-\tfrac12| = \tfrac12$ at every letter.
So $\ell$ takes only the values $a$ and $a+1$, and the mean $a+\tfrac12$ puts probability one half on each.
Conversely, such an $\ell$ has variance $\tfrac14$.

(iv) By (ii), $\dkl{\Qs_\varepsilon}{P}/(2\varepsilon^2)\to1/(4\mathrm{Var}_P(\ell))$, which by (iii) is at least one, with equality exactly when $\mathrm{Var}_P(\ell) = \tfrac14$.
\end{proof}

\subsection{Proof for Section~\ref{sec:learning}}

\begin{proof}[Proof of Proposition~\ref{prop:sampler}]
Under $R^{\otimes n}$ a type $t$ has probability $\Prob_P(\Pemp = t)e^{n\ip{g}{t}}$, so $\EV[X] = \sum_{t\in\cA}\Prob_P(\Pemp=t) = \Prob_P(\Pemp\in\cA)$ and $\EV[X_B] = \Prob_P(\Pemp\in B)$, while $\EV[X_B^2] = \sum_{t\in B}\Prob_P(\Pemp = t)e^{-n\ip{g}{t}}$.
Keep the single term $t = \Qs_{B,n}$ and use $\ip{g}{t} = \dkl{t}{P} - \dkl{t}{R}$ and the type bounds of the proof of Proposition~\ref{prop:bracket}: $\EV[X_B^2]\ge(n+1)^{-K}e^{-2n\dkl{\Qs_{B,n}}{P} + n\dkl{\Qs_{B,n}}{R}}$, and $\EV[X_B]\le(n+1)^Ke^{-n\dkl{\Qs_{B,n}}{P}}$.
Dividing gives the bound.
\end{proof}

\section{Enumerations}
\label{app:eval}

The probabilities, conditioned marginals and total correlations below are exact: each is a finite sum over the types of the lattice, carried out in floating point.
Only the importance-sampling estimates of Appendix~\ref{sec:eval-mixture-limit} are random, and they are reported as ratios to the exact value.

%
%
%
%
%
%
%
%
%
%
%
%
%
%
%
%
%
%
%
\subsection{The regime map}
\label{sec:eval-regime-map}

Figure~\ref{fig:regime-dial} shows $2954$ configurations of the $11$ families drawn above its plot, and Figure~\ref{fig:regime-signatures} the $1538$ of them on the shared range of sample sizes defined below.
Each family is a set $\cA$ of distributions on $\cX$, and the event is $\{\Pemp \in \cA\}$.
A configuration fixes the alphabet, the population, the sample size and the parameters of the set, which run over fixed grids.
The alphabets have $K = 3$ and $K = 4$ letters.
The populations are the uniform one and a skewed one, $P = (0.5, 0.3, 0.2)$ on three letters and $P = (0.4, 0.3, 0.2, 0.1)$ on four.
The loss is $\ell = (0, 1/(K-1), \ldots, 1)$, and $\mathrm{Var}_t(\ell) = \ip{(\ell - \ip{\ell}{t})^2}{t}$ is its variance under a distribution $t$.
Both figures plot a shared range of sample sizes, $n \in \{6,\allowbreak 8,\allowbreak 10,\allowbreak 12,\allowbreak 15,\allowbreak 18,\allowbreak 21,\allowbreak 24,\allowbreak 28,\allowbreak 32,\allowbreak 36,\allowbreak 40\}$ on three letters and the first eight of these, up to $n = 24$, on four.
Figure~\ref{fig:regime-dial} alone plots a larger range, $n \in \{45,\allowbreak 50,\allowbreak 55,\allowbreak 60,\allowbreak 70,\allowbreak 80,\allowbreak 90,\allowbreak 100\}$ on three letters and $\{28,\allowbreak 32,\allowbreak 36,\allowbreak 40,\allowbreak 45,\allowbreak 50,\allowbreak 60,\allowbreak 70,\allowbreak 80\}$ on four.
Unless a family names its population, its configurations run on both populations of each alphabet and on both ranges.
\begin{description}
\item[Half-space] $\{t : \ip{\ell}{t} \ge \ip{\ell}{P} + \varepsilon\}$ for $\varepsilon \in \{0.02, 0.05, 0.08, 0.12, 0.16, 0.20, 0.25, 0.30, 0.35, 0.40\}$; $399$ configurations on the shared range and $340$ on the larger.
\item[Linear family] $\{t : \ip{\ell}{t} = \ip{\ell}{P} + \varepsilon\}$ for $\varepsilon \in \{0.05, 0.12, 0.20, 0.30\}$, at the sample sizes where the level set contains a type; $21$ configurations on the shared range, all at $n \in \{10, 15, 40\}$, and $76$ on the larger, at every size on three letters and at $n \ge 40$ on four.
\item[Variance set] $\{t : \mathrm{Var}_t(\ell) \ge v\}$ for $v \in \{0.06, 0.09, 0.12, 0.15, 0.18, 0.21, 0.24\}$; $157$ configurations on the shared range and $131$ on the larger.
The set is convex, because the variance of a mixture is at least the mixture of the variances.
\item[Relative-entropy ball] $\{t : \dkl{t}{\pi} \le r\}$ for $r \in \{0.02, 0.05, 0.10, 0.20, 0.35\}$, centered at $\pi(x) = P(x-1)$, the population with its letters shifted cyclically by one, on the skewed populations; $87$ configurations on the shared range and $69$ on the larger.
\item[Ball centered on $P$] The same balls on a uniform population, where the shift leaves $P$ unchanged; $35$ configurations on the shared range and $7$ on the larger.
Those balls contain $P$ and are invariant under the full symmetric group.
\item[Face] $\{t : \supp t \subseteq F\}$ for $F$ the first $|F|$ letters and $1 \le |F| \le K - 1$; $95$ configurations on the shared range and $86$ on the larger.
\item[Union of $k < K$ half-spaces] $\bigcup_{j \le k}\{t : t_j \ge 1/K + \varepsilon\}$ for $2 \le k \le K - 1$ and $\varepsilon \in \{0.05, 0.10, 0.18, 0.28\}$, on the uniform population; $85$ configurations on the shared range and $93$ on the larger.
\item[Union of all $K$ half-spaces] The same union with $k = K$; $47$ configurations on the shared range and $52$ on the larger.
\item[Union of faces] $\{t : t_j = 0 \text{ for some } j \le m\}$ for $2 \le m \le K$; $93$ configurations on the shared range and $86$ on the larger.
\item[Two-sided tail] $\{t : |\ip{\ell}{t} - \ip{\ell}{P}| \ge \varepsilon\}$ for $\varepsilon \in \{0.05, 0.10, \ldots, 0.40\}$ with $\varepsilon < \min(\ip{\ell}{P}, 1 - \ip{\ell}{P})$, so that both tails are nonempty; $242$ configurations on the shared range and $238$ on the larger.
\item[Low-variance set] $\{t : \mathrm{Var}_t(\ell) \le v\}$ for $v \in \{0.01, 0.02, 0.03, 0.04, 0.06, 0.08, 0.10\}$ with $v < \mathrm{Var}_P(\ell)$, so that the set excludes $P$; $277$ configurations on the shared range and $238$ on the larger.
\end{description}
A configuration is kept when the event contains a type and has probability at most one half.
The exponent in the denominator of the share is then at least $\log 2$.
The rarest event has probability $10^{-19.1}$ on the shared range and $10^{-48.2}$ on the larger.

A type $t \in \cT_n$ has probability $\Prob_P(\Pemp = t) = \binom{n}{nt}\prod_x P(x)^{nt(x)}$, where $\binom{n}{nt} = n!/\prod_x (nt(x))!$ is the size of its type class and is computed through the logarithm of the gamma function.
The lattice $\cT_n$ has $\binom{n+K-1}{K-1}$ types, $5151$ at $n = 100$ on three letters and $91{,}881$ at $n = 80$ on four, so every quantity below is an exact sum.
The probability of the event is $\Prob_P(\Pemp\in\cA) = \sum_{t\in\cA\cap\cT_n}\Prob_P(\Pemp = t)$, accumulated in log space.
With the weights $w_t = \Prob_P(\Pemp = t)/\Prob_P(\Pemp\in\cA)$, the conditioned marginal is $\om = \sum_{t} w_t\,t$ and the marginal rate is $n\,\dkl{\om}{P}$.
The conditioned sample law $\muA$ is uniform on each type class, so $\Ent(\muA) = -\sum_t w_t\log w_t + \sum_t w_t\log\binom{n}{nt}$ and $\TC(\cA) = n\,\Ent(\om) - \Ent(\muA)$.
The three terms of \eqref{eq:lemma4} are computed separately.
The identity holds on every configuration to within $1.3\times10^{-13}$ nats on the shared range and $1.0\times10^{-12}$ nats on the larger.
For a half-space or a linear family the information projection is the exponential tilt $\Qs(x) \propto P(x)e^{\lambda\ell(x)}$ with $\lambda$ solving the constraint as an equality.
For a ball it is $\Qs(x) \propto P(x)^{1/(1+b)}\pi(x)^{b/(1+b)}$ with $b$ solving $\dkl{\Qs}{\pi} = r$.
For a face it is $P$ conditioned on $F$.
For a union of half-spaces and for a two-sided tail it is the projection onto the piece of least relative entropy from $P$.
In each case $\Qs = P$ when $P$ lies in the set.
For the variance set the most probable type stands in for the projection.

Each record stores the probability, the marginal rate, the dependence, the share, the residual of \eqref{eq:lemma4} and the rate of the most probable type with the bracket of Proposition~\ref{prop:bracket}.
The record of a two-sided tail also stores the rate of its projection.
Where a projection is computed or the most probable type stands in for it, the record also stores the divergence of the marginal from that distribution and the defect $\Delta(\om)$ of \eqref{eq:defect} with that distribution in place of $\Qs$.
On $16$ variance-set configurations of four letters the most probable type has a zero coordinate where the marginal is positive, so the divergence is $+\infty$ and the defect $-\infty$.
The union of faces contains no distribution of full support, so its projection, $P$ conditioned on its most probable face, lacks full support too.
The rate of that face is recorded.
The divergence from the projection and the defect are left empty.
The low-variance set has a projection of full support but not in closed form.
The sweep does not compute it, so its divergence from the projection and its defect are left empty too.
Figures~\ref{fig:regime-dial} and~\ref{fig:regime-signatures} plot the share and the two terms against $-\log_{10}\Prob_P(\Pemp\in\cA)$.

\begin{figure}[t]
\centering
\includegraphics[width=\linewidth]{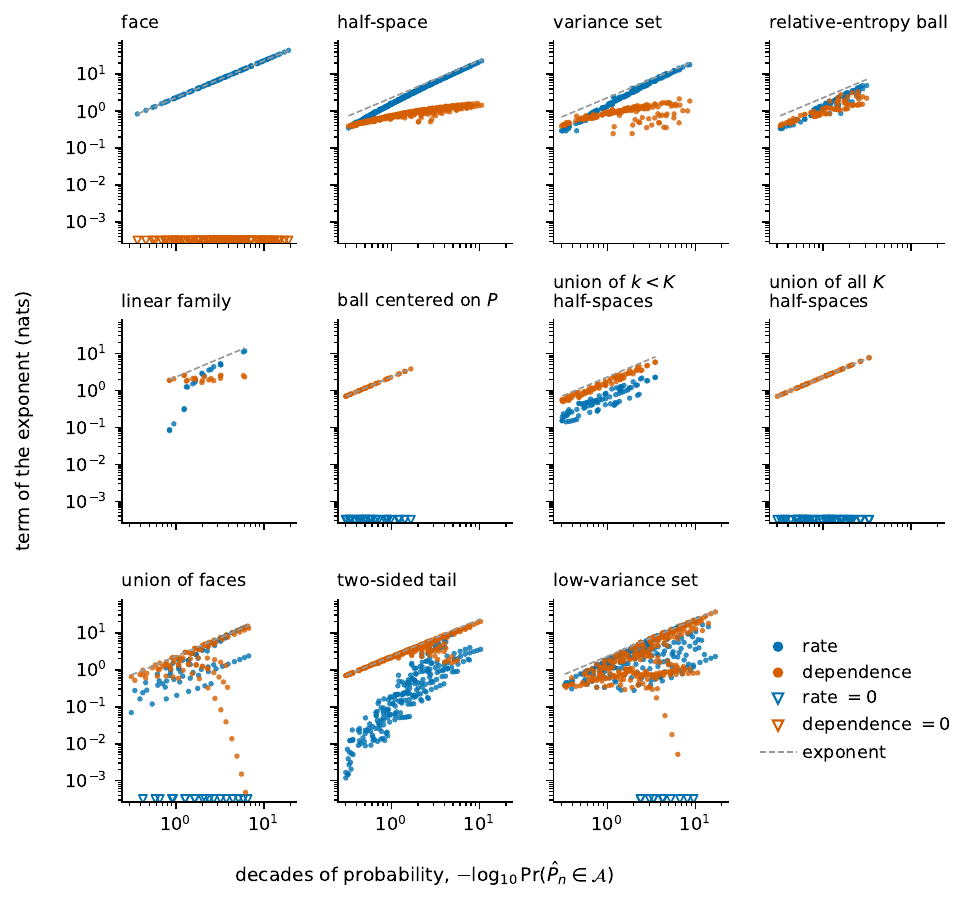}
\caption{\textbf{On a convex set that excludes $P$ the dependence grows on a shallower slope than the exponent, while on a union of separated modes of comparable probability it grows in proportion to the exponent.}
Both terms of \eqref{eq:lemma4} against the rarity of the event, one panel per family of Figure~\ref{fig:regime-dial}, at sample sizes $n \le 40$ on three letters and $n \le 24$ on four; the dashed line is the exponent.
Both axes are logarithmic, and open triangles at the foot of a panel mark a term that vanishes.
On the unions of faces of a skewed population one face comes to dominate, so the dependence eventually falls.}
\label{fig:regime-signatures}
\end{figure}

The two rows of glyphs at the top of Figure~\ref{fig:regime-dial}, the convex families in the first, draw one example set of each family's form on three letters, with $\ell = (0, 1/2, 1)$.
The convex families other than the ball centered on $P$ are drawn at $P = (0.5, 0.3, 0.2)$.
They are the face $\{\supp t \subseteq \{1, 2\}\}$, the half-space $\{\ip{\ell}{t} \ge 1/2\}$, the variance set $\{\mathrm{Var}_t(\ell) \ge 3/16\}$, the ball $\{\dkl{t}{(0.2, 0.3, 0.5)} \le 0.15\}$ and the linear family $\{\ip{\ell}{t} = 5/8\}$.
This ball is one of the balls of Appendix~\ref{sec:eval-curvature}, centered at $P$ with its first and last letters exchanged.
On the skewed population of three letters the balls of the sweep are centered at $(0.2, 0.5, 0.3)$.
The ball centered on $P$ is $\{\dkl{t}{P} \le 0.03\}$ at the uniform $P$.
The unions over $j \le k$ of the half-spaces $\{t_j \ge 1/3 + 1/4\}$ for $k = 2$ and $k = 3$ are drawn at the uniform population.
The union of all three faces, which is the boundary of the simplex, and the two-sided tail $\{|\ip{\ell}{t} - 1/2| \ge 1/4\}$ are also drawn at the uniform population.
The low-variance set $\{\mathrm{Var}_t(\ell) \le 0.05\}$ is drawn at $P = (0.5, 0.3, 0.2)$.
The crosses mark the information projection onto each piece.
On the face, the half-space, the linear family, the two balls, the unions of half-spaces and the union of faces they are the projections of the preceding paragraph.
On the variance set the cross marks the distribution of least relative entropy from $P$ in the set on a grid of $400$ points per side.
On the two-sided tail and the low-variance set the crosses mark the grid distribution of least relative entropy from $P$ in each piece, on the drawing grid of $260$ points per side.

On the shared range the computed dependence is below $10^{-13}$ on the $95$ faces and at least $4.8\times10^{-4}$ on every other configuration.
Off the faces the dependence falls below $0.24$ nats only on $11$ configurations, all on the skewed population of three letters.
They are unions of two faces and low-variance sets, where one piece of the set dominates.
The computed marginal rate is below $10^{-13}$ on $114$ configurations and at least $1.1\times10^{-3}$ on every other.
The $114$ are the $47$ unions of all $K$ half-spaces, the $35$ balls centered on $P$, the $19$ unions of all $K$ faces and the $13$ low-variance sets whose types are the vertices of the simplex, all on a uniform population.
They are exactly the configurations whose types form a set invariant under a transitive group of letter permutations that fixes the population.

On the shared range and a uniform population the dependence share of the unions of fewer than $K$ half-spaces or faces and of the two-sided tails is at least $0.57$.
The dependence share of the two-sided tails alone is at least $0.70$.
On a skewed population the most probable face comes to dominate a union of faces, so the share of the union falls toward zero.
On the union of two faces of three letters it reaches $3.4\times10^{-5}$ at $n = 40$.
The $399$ half-spaces lie below the limiting curve of Proposition~\ref{prop:halfspace-gauss} or at most $0.007$ above it.
Their median distance from it falls from $0.038$ at $n = 6$ to $0.007$ at $n = 40$.
The $20$ more than $0.05$ below it have $n\le 12$ and a threshold within $0.3$ of the largest value of $\ell$, where the Gaussian approximation to the tail of a bounded statistic is poor.

%
%
%
%
%
%
%
%
%
%

On the larger range the computed dependence is below $10^{-12}$ on the $86$ faces and positive on every other configuration.
Its smallest value off the faces is below $10^{-11}$ nats, a share below $10^{-12}$ of the exponent.
It occurs on the union of two faces of the skewed three-letter population at $n = 100$, where one face has come to dominate.
The computed marginal rate is below $10^{-12}$ on $76$ configurations and at least $1.4\times10^{-3}$ on every other.
The $76$ are the $52$ unions of all $K$ half-spaces, the $7$ balls centered on $P$ and the $17$ unions of all $K$ faces, all on a uniform population.
They are exactly the configurations invariant under a transitive group of letter permutations that fixes the population.
No low-variance set is among them, because none reduces to the vertices of the simplex at these sizes.

On the larger range and a uniform population the dependence share is at least $0.62$ on the unions of fewer than $K$ half-spaces, at least $0.78$ on the unions of fewer than $K$ faces and at least $0.76$ on the two-sided tails.
On a skewed population the share of a union of faces keeps falling.
For three faces of three letters it falls from $0.037$ at $n = 40$ to $2.9\times10^{-5}$ at $n = 100$.
For the unions of two, three and four faces of four letters it lies between $8.1\times10^{-5}$ and $9.8\times10^{-4}$ at $n = 80$.
The $340$ half-spaces lie below the limiting curve of Proposition~\ref{prop:halfspace-gauss} or at most $0.007$ above it.
None falls more than $0.02$ below the curve.
Their median distance from it falls from $0.007$ at $n = 28$ to $0.003$ at $n = 100$.
At $n = 100$ on three letters the $18$ configurations of the variance sets, the relative-entropy balls and the linear families have a dependence share between $0.06$ and $0.80$.
Of these, $16$ have a counterpart at $n = 40$ with the same set and population.
All $16$ have a lower share at $n = 100$ than their counterparts.

%
%
%
%
%
%
%
%
%
\subsection{A union's exponent against its number of modes, and the weights of equal-rate modes}
\label{sec:eval-union-tax}

Figure~\ref{fig:eval-unions}(a) takes a six-letter uniform population at $n=16$ and the union of the first $k$ half-spaces $\{t_j\ge 1/6+0.25\}$ for $k=1,\ldots,6$.
The group $S_k\times S_{6-k}$ permutes the first $k$ letters among themselves and the others among themselves and fixes $P$ and the union.
By Proposition~\ref{prop:symmetry} the marginal rate per sample is therefore the relative entropy of the two orbit masses of $\om$ from those of $P$, which are $k/6$ and $1-k/6$.
As $k$ grows the union becomes more probable and its marginal rate shrinks, to zero at $k=6$, where the group acts transitively and the whole exponent is dependence.

In panels (b) and (c) the union has two modes of equal rate and unequal prefactors on three letters: $\{t_1\ge 3/4\}$ and $\{t_2\ge 1/2\}$, whose thresholds are attainable on the lattice at every $n$ divisible by four.
The population is $P = (0.5, p, 0.5-p)$ with the $p\approx 0.2601$ at which the two projections have the same relative entropy, $0.1308$.
The pieces do not overlap at any $n$ enumerated.
The lattice form of the Bahadur--Rao constant, $1/((1-e^{-\lambda_j})\sigma_j)$, predicts the weight ratio $1.123$ and the first-mode weight $0.529$.
The enumerated ratio falls from $1.128$ at $n=12$ to $1.123$ at $n=768$, so the modes do not approach equal weights.
The dependence per sample decreases toward $0.0996$, the Jensen--Shannon divergence of the two projections at the prefactor weights and the limit of Corollary~\ref{cor:mixture-limit}.
The Jensen--Shannon divergence of the piece marginals is within $10^{-3}$ of that limit at $n=768$.

\begin{figure}[htbp]
\centering
\includegraphics[width=\linewidth]{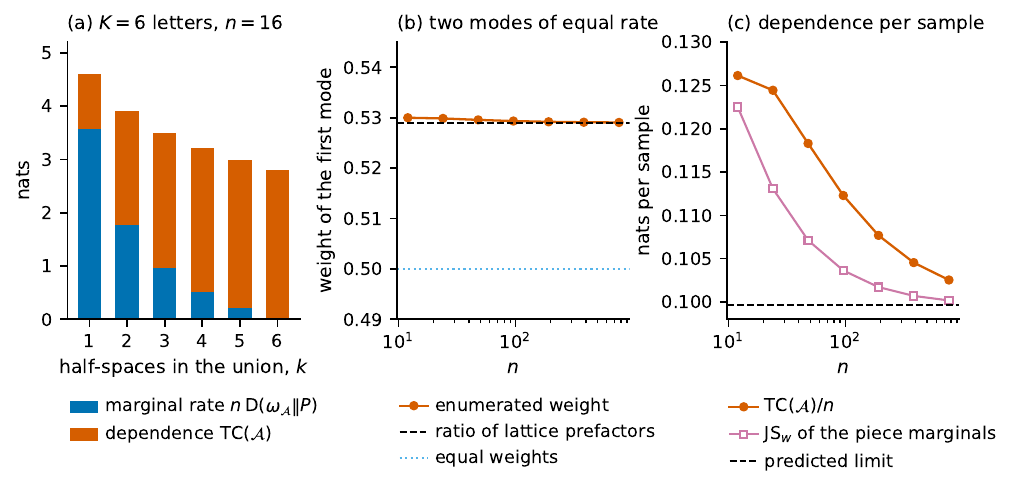}
\caption{\textbf{The dependence share of a union of half-spaces grows as its modes multiply.}
Modes of equal rate weigh in the ratio of their lattice prefactors.
(a) The exponent of the union of the first $k$ half-spaces $\{t_j\ge 1/6+0.25\}$ on six uniform letters at $n = 16$, split into the marginal rate and the dependence.
(b) The conditional weight of the first of two modes of equal rate on three letters, against the ratio of their lattice Bahadur--Rao constants and against equal weights.
(c) The dependence per sample of that two-mode union and the Jensen--Shannon divergence of its piece marginals, against the limit of Corollary~\ref{cor:mixture-limit} at the prefactor weights.}
\label{fig:eval-unions}
\end{figure}

%
%
%
%
%
%
%
%
\subsection{The mixture limit and the two samplers}
\label{sec:eval-mixture-limit}

On the four-letter uniform population the union of the first two half-spaces $\{t_j\ge 1/2\}$ has two modes with projections $\Qs_1 = (0.500, 0.167, 0.167, 0.167)$ and its permutation, of common rate $\dkl{\Qs}{P} = 0.1438$.
Their equal-weight mixture $\bar Q$ has $\dkl{\bar Q}{P} = 0.0566$, so Corollary~\ref{cor:mixture-limit} gives the limiting dependence share $1 - 0.0566/0.1438 = 0.606$.
Figure~\ref{fig:eval-mixture} follows $n$ from $8$ to $100$ against the half-space $\{t_1\ge 1/2\}$ on the same population.
The $\ell_1$ distance from $\om$ to $\bar Q$ falls as $n^{-0.86}$ on a least-squares fit over $n\ge 20$, the exponent at which the conditioned marginal of the half-space approaches its projection, while the distance to a single mode stays near $0.34$.
Under the hypothesis of Corollary~\ref{cor:mixture-limit}, Pinsker's inequality bounds the $\ell_1$ distance from a convex piece's marginal to its projection by $\sqrt{2(c + K\log(n+1))/n}$, a guaranteed rate of one half up to the logarithm, which the fitted $0.86$ exceeds.
The union's dependence share rises to $0.637$ at $n=28$ and is $0.629$ at $n=100$, while the half-space's falls from $0.32$ to $0.10$.

\begin{figure}[htbp]
\centering
\includegraphics[width=\linewidth]{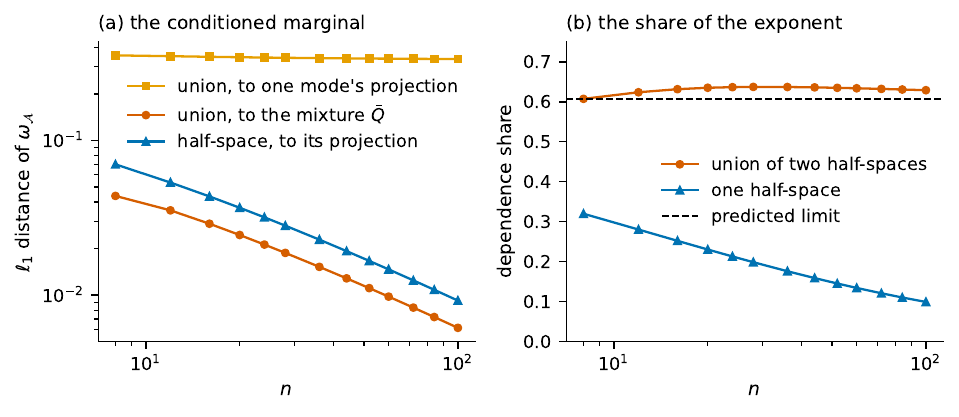}
\caption{\textbf{The conditioned marginal of a two-mode union converges to the mixture of its modes as fast as a half-space's converges to its projection.}
The union's dependence share stays near its limit.
Four-letter uniform population; the union of $\{t_1\ge 1/2\}$ and $\{t_2\ge 1/2\}$ against the half-space $\{t_1\ge 1/2\}$.
(a) The $\ell_1$ distance of the conditioned marginal to one mode's projection and to the mixture $\bar Q$ of the two, and for the half-space to its projection.
(b) The dependence share, with the limit $0.606$ of Corollary~\ref{cor:mixture-limit} for the union.}
\label{fig:eval-mixture}
\end{figure}

The second measurement, Figure~\ref{fig:learning}(d), estimates the probability of the four-mode symmetric union $\bigcup_j\{t_j\ge 1/2\}$ by importance sampling with $40{,}000$ draws per estimate.
One sampler draws from the projection $\Qs$ of the first mode, as the change-of-measure form \eqref{eq:sharp-sanov} prescribes; the other draws from the equal mixture of the four mode projections and weights by the mixture likelihood ratio.
The exact value comes from the lattice for $n\le 60$ and from inclusion--exclusion over the joint multinomial tails beyond.
From $n=40$ on, the single-projection estimate sits at one quarter of the exact value, the probability of the sampler's own mode, with a reported relative standard deviation near two.
At $n=20$ and $n=30$ it lands below and above the truth with relative standard deviations of $27$ and $115$, the signature of an occasional draw from another mode with an enormous weight.
The mixture sampler is within two percent of the exact value at every $n$ up to $160$, where the exact probability is $4\times10^{-11}$, with a relative standard deviation growing slowly from $1.1$ to $2.7$.

%
%
%
%
%
%
%
%
%
%
%
%
%
%
%
%
%
\subsection{The curvature gap}
\label{sec:eval-curvature}

For the balls $\{t : \dkl{t}{\pi}\le r\}$ around $\pi = (0.2, 0.3, 0.5)$ under $P = (0.5, 0.3, 0.2)$, six radii from $0.02$ to $0.25$ are enumerated against their tangent half-spaces $\cH$ at the projection (Proposition~\ref{prop:tangent}).
The whole lattice is enumerated at eight sample sizes from $12$ to $1536$.
The other sample sizes, up to $10^6$, are summed through \eqref{eq:gap-tilted} over the band of types with $0\le n\ip{g}{t-\Qs}\le 45$, whose omitted types change each probability by a relative amount below $10^{-15}$.

The balls fall under Proposition~\ref{prop:curvature-limit}(ii).
Here $\log(\pi/P) = \log(2.5)\,(-1, 0, 1)$ and the projection $\Qs$ is proportional to $P^{1/(1+\beta)}\pi^{\beta/(1+\beta)}$, so the values of $g$ are $g(1) + \eta\,(0, 1, 2)$ with $\eta = \beta\log(2.5)/(1+\beta)$.
The Hessian of the boundary is $M = \beta\,\Sigma_\zeta^{-1}$, so $q(\zeta) = \tfrac{\beta}{2}\chi^2_1$, the limit the gap would have without the lattice is $\tfrac12\log(1+\beta)$, and
\[
  F(x) = -\log\EV\exp\bigl(-\eta\,\lceil\beta\chi^2_1/(2\eta) - x\rceil\bigr)
\]
with $\chi^2_1$ a chi-square variable with one degree of freedom.
For the balls of radius $0.05$ and $0.15$ the range of $F$ has width $\eta$ and contains $\tfrac12\log(1+\beta)$.
\begin{center}
\begin{tabular}{@{}ccccc@{}}
\toprule
$r$ & $\beta$ & $\eta$ & range of $F$ & $\tfrac12\log(1+\beta)$ \\
\midrule
$0.05$ & $1.357$ & $0.528$ & $[0.268, 0.795]$ & $0.4288$ \\
$0.15$ & $0.369$ & $0.247$ & $[0.084, 0.331]$ & $0.1570$ \\
\bottomrule
\end{tabular}
\end{center}
At every radius $|\tau_n - F(x_n)|$ falls from between $6.4\times10^{-3}$ and $1.8\times10^{-2}$ at $n = 10^3$ to between $2.7\times10^{-5}$ and $7.6\times10^{-4}$ at $n = 10^6$.
The distance also depends on the phase, and it need not fall at every step: at $r = 0.05$ it is $2.1\times10^{-4}$ at $n = 10^5$ and $2.8\times10^{-4}$ at $n = 10^6$.
At $n = 10^6$ the gaps of the balls of radius $0.05$ and $0.15$ both lie within $3\times10^{-4}$ of $F(x_n)$.

At $r = 0.05$ the gap itself keeps moving across the range of $F$ as the phase $x_n$ moves.
\begin{center}
\begin{tabular}{@{}lcccc@{}}
\toprule
$n$ & $10^3$ & $10^4$ & $10^5$ & $10^6$ \\
\midrule
$x_n$ & $0.92$ & $0.23$ & $0.27$ & $0.65$ \\
$\tau_n$ & $0.297$ & $0.506$ & $0.487$ & $0.345$ \\
\bottomrule
\end{tabular}
\end{center}
Figure~\ref{fig:eval-curvature}(a) plots the exact gap of that ball against the phase at every $n$ from $2000$ to $4000$, and the points lie within $0.034$ of $F$, with a median distance of $0.005$.

Averaged over the phase, the two leading terms in the proof of Proposition~\ref{prop:curvature-limit} have ratio $\EV e^{-q(\zeta)}$, the limit without the lattice, because $x + \lceil y - x\rceil - y$ is uniform on $[0,1)$ when $x$ is.
In the same sense the half-space's lattice constant oscillates about its value without the lattice \cite{BahadurRao1960}.
The average of the gap over consecutive sample sizes is a different quantity.
Over the thousand sample sizes $n = 20000,\ldots,20999$ the average of the exact gap is close to the average of $F$ at the same phases and to the phase average $\int_0^1 F$.
\begin{center}
\begin{tabular}{@{}cccc@{}}
\toprule
$r$ & average of $\tau_n$ & average of $F(x_n)$ & $\int_0^1 F$ \\
\midrule
$0.05$ & $0.4192$ & $0.4190$ & $0.4185$ \\
$0.15$ & $0.1550$ & $0.1548$ & $0.1548$ \\
\bottomrule
\end{tabular}
\end{center}
At both radii the three averages lie below $\tfrac12\log(1+\beta)$, which is $0.4288$ and $0.1570$.
At every radius the phase average $\int_0^1 F$ lies below $\tfrac12\log(1+\beta)$, by $0.017$ at $r = 0.02$ and by less than $10^{-3}$ at $r = 0.20$ and $0.25$.

A control with incommensurable log-ratios, $\pi' = (0.2, 0.5, 0.3)$, falls under case (i).
The differences of the values of $g$ are in the ratio of $\log(4/15)$ to $\log(10/9)$, which is irrational because no nonzero integer powers of $4/15$ and $10/9$ are equal.
At $r = 0.05$ the limit is $\tfrac12\log(1+\beta') = 0.3419$.
Over windows of $40$ consecutive sample sizes the largest deviation of the exact gap from that limit shrinks as $n$ grows.
\begin{center}
\begin{tabular}{@{}lccccc@{}}
\toprule
window start & $3000$ & $10^4$ & $3\times10^4$ & $10^5$ & $3\times10^5$ \\
\midrule
largest deviation & $0.048$ & $0.013$ & $0.0062$ & $0.0055$ & $0.0048$ \\
\bottomrule
\end{tabular}
\end{center}
From $n = 3\times10^4$ on the window means stay within $2\times10^{-4}$ of the limit (Figure~\ref{fig:eval-curvature}(b)).
The approach is slow: the largest deviation falls by a factor of $10$ between $n = 3000$ and $n = 3\times10^5$.
Panel (b) of Figure~\ref{fig:eval-curvature} follows both balls of radius $0.05$ at $41$ sample sizes from $10^2$ to $10^6$.

\begin{figure}[htbp]
\centering
\includegraphics[width=\linewidth]{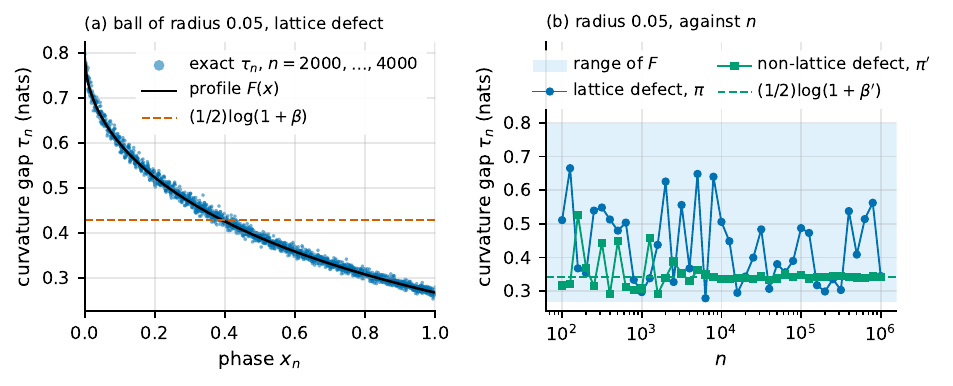}
\caption{\textbf{The curvature gap of a ball whose defect is a lattice statistic lies on a profile of the phase and has no limit, while the ball with incommensurable log-ratios converges.}
(a) The exact gap of the ball of radius $0.05$ at every $n$ from $2000$ to $4000$, against the phase $x_n$; the curve is $F$, and the dashed line is $\tfrac12\log(1+\beta)$.
(b) The exact gap against $n$ from $10^2$ to $10^6$ for the same ball, with the range of $F$ shaded, and for the ball around $\pi' = (0.2, 0.5, 0.3)$, which approaches its dashed limit.}
\label{fig:eval-curvature}
\end{figure}

The smallest defect over the types in the ball of radius $0.05$ falls from $0.019$ at $n=12$ to $0.00024$ at $n=1536$ as the lattice refines toward $\Qs$, so the zero set of the defect locates the exposed face while the gap measures the curvature.
The ball's dependence share falls with $n$ alongside the half-space's, from $0.40$ and $0.30$ at $n=12$ to $0.020$ and $0.019$ at $n=1536$, both logarithmic corrections to an exponent that grows linearly.

%
%
%
%
%
%
%
%
\subsection{Uniform convergence and the empirical-risk minimizer}
\label{sec:eval-classification}

The alphabet of Figure~\ref{fig:learning}(a) is $\cX = \{1,2,3\}\times\{0,1\}$, six letters, with the inputs uniform and $\Prob(Y=1\mid X=x) = 1/2 + s\,\sigma(x)$ for $\sigma = (+1,-1,+1)$ and a signal $s\in\{0, 0.1, 0.2, 0.3\}$.
The class is the $2^3$ labelings of the inputs, with the nested subclasses of the $k\in\{1,2,4,8\}$ labelings that assign label $0$ to the first $3-\log_2 k$ inputs.
The events computed at $\varepsilon = 0.2$ are the uniform-convergence event $U_\varepsilon$ over the subclass and the event that the empirical-risk minimizer over the subclass has population risk at least $\varepsilon$ above its empirical risk.
At zero signal every labeling has population risk one half, and a flip of every label maps each event onto the other while fixing $P$, so the two have the same probability and the same share.
The group of input permutations and per-input label flips then acts transitively on the six letters, fixes $P$ and permutes the full class, so by Proposition~\ref{prop:symmetry} the dependence share of both events over the full class is one.
The panel shows $n=24$; at $n=12$ and $n=18$ each share lies within $0.08$ of its value at $n=24$.
At every signal and at all three sample sizes the dependence share of either event over the full class stays within $0.05$ of one, while over the single labeling of $k = 1$ it is $0.20$ to $0.32$.
Over the full class the minimizer's gap event is rarer than the uniform-convergence event by a factor that grows with the signal, reaching four at $s=0.3$ and $n=24$.

The population of Figure~\ref{fig:learning}(b) is a real dataset: the $150$ iris flowers placed in the nine cells of petal-length tercile by species, whose empirical distribution is $P$.
Four of the nine cells are empty.
The class of all $27$ maps from tercile to species has Bayes risk $0.047$.
Sampling from this $P$ is sampling with replacement from the dataset.
The three events are computed at $n = 10, 12, 14$ and $\varepsilon = 0.15, 0.25$: the uniform-convergence event over the $27$ maps, the minimizer's gap event, and the deviation of the single best map, a half-space.
Over these sample sizes and levels the dependence shares of the three events lie in ranges that do not overlap.
\begin{center}
\begin{tabular}{@{}lcc@{}}
\toprule
event & smallest share & largest share \\
\midrule
union over the $27$ maps & $0.90$ & $0.99$ \\
minimizer's gap event & $0.63$ & $0.81$ \\
best single map & $0.16$ & $0.33$ \\
\bottomrule
\end{tabular}
\end{center}

%
%
%
%
%
%
%
%
\subsection{Calibrating the relative-entropy ball}
\label{sec:eval-dro-ball}

For the coverage event $\{\dkl{\Pemp}{P}\le r\}$ at level $1-\alpha = 0.95$, three radii are computed on uniform populations of three to five letters and on the skewed population $(0.6,0.3,0.1)$, at $n = 10, 20, 40, 80$.
The exact radius is the smallest attainable $r$ whose coverage reaches $0.95$.
The likelihood-ratio radius is $\chi^2_{K-1,0.95}/(2n)$.
The method-of-types radius is $(\log(1/\alpha) + K\log(n+1))/n$.
The likelihood-ratio radius undercovers at small $n$, at $0.92$ on four and on five letters at $n=10$, and is within one percent of the exact coverage by $n=80$.
The method-of-types radius is two and a half to five and a half times the exact one, the ratio growing with $n$ as its $K\log(n+1)$ term does, and covers with probability at least $0.9999$ at every $n$.
Figure~\ref{fig:learning}(c) shows the dependence share of the failure event $\{\dkl{\Pemp}{P} > r\}$ at the exact radius.
It is one on every uniform population, the symmetric-group case of Proposition~\ref{prop:symmetry} with a continuum of modes, while on the skewed population it rises from $0.88$ at $n=10$ to $0.999$ at $n=80$, as the conditioned marginal of the failure event approaches the population.

%
%
%
%
%
%
%
%
%
%
\subsection{The Hoeffding ratio and the two-coordinate dependence}
\label{sec:eval-classical}

\begin{samepage}
For the loss $\ell = (0,\tfrac12,1)$ under $P = (0.5,0.3,0.2)$, whose variance is $0.1525$, the table below gives the ratio of the Chernoff exponent of the half-space $\{\ip{\ell}{t}\ge\ip{\ell}{P}+\varepsilon\}$ to Hoeffding's exponent, $\dkl{\Qs}{P}/(2\varepsilon^2)$, at four thresholds.
\begin{center}
\begin{tabular}{@{}lcccc@{}}
\toprule
$\varepsilon$ & $0.25$ & $0.15$ & $0.10$ & $0.05$ \\
\midrule
$\dkl{\Qs}{P}/(2\varepsilon^2)$ & $1.53$ & $1.55$ & $1.58$ & $1.60$ \\
\bottomrule
\end{tabular}
\end{center}
\end{samepage}
Over these thresholds the ratio rises toward its limit $1/(4\cdot 0.1525)\approx 1.64$ (Lemma~\ref{lem:hoeffding-ratio}).
On the half-space $\{\ip{\ell}{t}\ge\tfrac12\}$ of the same example the total correlation of two coordinates falls from $2.8\times10^{-3}$ at $n=12$ to $1.3\times10^{-5}$ at $n=192$.
Over the same range $\TC(\cA)$ rises from $0.704$ to $1.637$.

\end{document}